\documentclass[final]{ectaart}

\RequirePackage[OT1]{fontenc}
\RequirePackage{graphicx}
\RequirePackage{amsthm}
\RequirePackage[cmex10]{amsmath}
\RequirePackage{natbib}
\RequirePackage[colorlinks,citecolor=blue,urlcolor=blue]{hyperref}
\usepackage{color}
\usepackage{bm}
\usepackage{graphicx}
\usepackage{epstopdf}
\usepackage{array,epsfig}
\usepackage{dsfont}

\usepackage{amsmath}
\usepackage{amssymb}
\usepackage{amsfonts}
\usepackage{multirow}
\usepackage{amsthm}

\newtheorem{thm}{Theorem}

\newtheorem{lem}{Lemma}

\newtheorem{ass}{Assumption}

\theoremstyle{definition}

\newtheorem{rem}{Remark}

\newcommand{\R}{\mathbb{R}}
\newcommand{\E}{\mathbb{E}}
\newcommand{\Pro}{\mathcal{P}}

\allowdisplaybreaks

\begin{document}

\begin{frontmatter}
\title{Double Cross Validation for the Number of  Factors in Approximate Factor Models}
\runtitle{Double Cross Validation for the Number of Factors in Approximate Factor Models}
\begin{aug}
\author{\fnms{Xianli} \snm{Zeng}\thanksref{a}\ead[label=author1]{stazeng@nus.edu.sg}}
\author{\fnms{Yingcun} \snm{Xia}\thanksref{b}\ead[label=author2]{staxyc@nus.edu.sg}}
\author{\fnms{Linjun} \snm{Zhang}\thanksref{c}\ead[label=author3]{linjunz@wharton.upenn.edu}}

\address[a]{Department of Statistics and Applied Probability, National University of Singapore, Singapore 117546; \printead{author1}}
\address[b]{Department of Statistics and Applied Probability, National University of Singapore, Singapore 117546; \printead{author2}}
\address[c]{Department of Statistics, the Wharton School, University of Pennsylvania, Philadelphia, PA 19104. \printead{author3}}
\runauthor{X. Zeng, Y. Xia and L. Zhang}
\end{aug}

\begin{abstract}
Determining the number of factors is essential to factor analysis.
In this paper, we propose  {an  efficient cross validation (CV)} method to determine the number of factors in approximate factor models.
The method applies 
 CV
 twice, first along the directions of observations and then variables, and hence is referred to hereafter as  double cross-validation (DCV).
Unlike most CV methods, which are prone to overfitting, the DCV is statistically consistent in determining the number of factors when both dimension of variables and sample size are sufficiently  large.
Simulation studies show that DCV has outstanding performance in comparison to existing methods in selecting the number of factors, especially when the idiosyncratic error has heteroscedasticity, or heavy tail, or relatively {large} variance.
\end{abstract}


\end{frontmatter}

\baselineskip1.8em

\section{introduction}

Factor models are a special kind of latent-variable models that are widely used in economics and other disciplines of research.
Due to the ubiquitous dependence across high-dimensional economic variables, it is appealing for explaining the variation of the high-dimensional economic measurements from only a small number of latent common factors.
Well-known examples include the arbitrage pricing theory \citep{Ross1976}, rank of demand systems \citep{Lewbel1991}, multiple-factor models \citep{FF1993}, components analysis in economic activities  \citep{GH1999}, and diffusion index \citep{SW2002}.
Successful applications rely heavily on the ability to correctly specify the number of factors, which is usually unknown in practice. Therefore, determining the number of factors is an essential step in applying factor models.

 One approach to this end is to utilize the eigen-structure of data matrix; see for example  \cite{F1977}, \cite{Schott1994},  \cite{YW2003}, \cite{Onatski2010} and \cite{LL2016}, amongst  others.
However, these methods require strong assumptions on the separation of eigenvalues,  and thus most of them are only applicable to fixed dimensional data.
More popular approaches are cross-validation (CV) and information criteria (IC).
{For factor model, several cross-validation methods have been proposed}; see, for example, \cite{Wold1978} and \cite{EK1982}, amongst  others.
\cite{BKSK2008} comprehensively reviewed these CV-based approaches, and concluded that most were not statistically consistent.
In addition, they recommended an alternative method based on the expectation maximization (EM) algorithm, which turned out to be extremely  computationally expensive.
As an alternative, the ICs are  more computationally effective and consistency can be guaranteed. As far as we know, \cite{CD1997} was the first paper that used ICs to determine the  number of factors.
Subsequently, other methods based on the Akaike information criterion (AIC) and Bayesian information criterion (BIC) \citep{SW1998,FLHR2000, BaiNg2002, LLS2017}  were proposed for different settings of sample size $n$ and dimension of variable, $p$.
When the variance of the idiosyncratic component is relatively small, these methods are able to provide very efficient estimates of the number of factors \citep{BaiNg2002}. 
However, the IC-based approaches are usually not fully data-driven since their penalties depend on predetermined tuning parameters.
Even though several guides are proposed to mitigate the effect of  tuning parameters, the performances of these methods are not stable when the signal-to-noise ratio is relatively small \citep{Onatski2010}, i.e., the variation of the idiosyncratic component is relatively large, or when the idiosyncratic component has heavy tails in distribution.
Unfortunately, both of which are common in financial data.

In this paper, we propose to estimate the number of factors by a computationally-efficient CV method.
Note that the conventional CV methods, e.g., those used in the linear regression, only leave one or several \emph{observations} out.
In contrast, factor models involve a matrix, where both observations (or rows) and variables (or columns) play similar roles mathematically in the analysis.
Thus, we propose a double cross-validation (DCV) method that applies CV first to the \emph{observations} and then to the \emph{variables} in the {observations}.
Theoretically, we show that the method is consistent for high-dimensional data and allows dependency between the idiosyncratic errors.
The method  is thus applicable to approximate factor models.
In addition, computationally, the second CV can be easily calculated by the linear regression error using the full data \citep{Shao1993}, while the first CV can be done in a $K$-fold manner.
Simulation studies show that the proposed approach performs satisfactorily even when the idiosyncratic  error has relatively big variance or homoscedasticity, or heavy tail,  which are especially relevant for economic data.

The rest of this paper is organized as follows. Section~\ref{sec:method} reviews the basic notation of factor analysis.
Section~\ref{DCVsection} presents the DCV approach.
Section~\ref{sec:theory} establishes statistical consistency of the approach, while the proofs of those properties are given in the Appendix.
Sections~\ref{sec:simulation}--\ref{sec:application} present a set of simulation studies and  an empirical application to assess the finite sample performance of DCV. Concluding remarks are given in Section~\ref{sec:discussion}.

\section{Factor model and its  factors and  loadings}\label{sec:method}

For a $p$-dimensional random vector ${\bm x}=(X_1, X_2,...,X_p)^\top $, the factor model assumes the following \emph{bilinear} structure
\begin{equation*}
{\bm x}={\bm L^0{\bm f}^0}+{\bm e},
\end{equation*}
where ${\bm L}^0\in\R^{p\times d_0}$ is the loading matrix, ${\bm f}^0=(f_{1}, f_{2}, ... , f_{d_0})^\top $ is  the vector of $ \ d_0 \  $ common factors. Together ${\bm L^0{\bm f}^0}$ is called the common component of the data,  and ${\bm e}=(e_{1},e_{2},...,e_{p})^\top $ the idiosyncratic component. Let $\boldsymbol\Sigma_{\bm e}=\E[\bm e\bm e^\top]$, $\boldsymbol\Sigma_{\bm f_0}=\E[{\bm f_0\bm f_0^\top}]$, and $\boldsymbol\Sigma_{\bm x}=\E[\bm x\bm x^\top]$. Under this model, the variation of the $p$-dimensional variable is mainly generated by a small number of common factors.
In the statistical literature, the idiosyncratic component are usually assumed to be independent of the common component, and have diagonal covariance matrix, in which case {$\boldsymbol\Sigma_{\bm x}={\bm L}^0\boldsymbol\Sigma_{\bm f_0}{{\bm L}^{0}}^\top +\boldsymbol\Sigma_{\bm e}$}. However, such independence assumption and diagonal structure are usually not realistic in financial data or macroeconomic data where factor analysis are often applied. In this paper, we consider the approximate factor model where the idiosyncratic components can be weakly dependent.

Let  ${\bm X}=({\bm x}^\top _1,{\bm x}^\top _2,...,{\bm x}^\top _n)^\top $ be random samples of  $n$ observations on $ p$ variables. We have the following matrix form for the factor model:
\begin{equation}  \label{FM}
{\bm X}={\bm F}^0{{\bm L}^{0}}^\top +{\bm E},
\end{equation}
with ${\bm F}^0=({\bm f}_1^0,{\bm f}_2^0,...,{\bm f}_n^0)^\top $ and  ${\bm E}=({\bm e}^\top _1,{\bm e}^\top _2,...,{\bm e}^\top _n)^\top $.
According to the eigen-decomposition, ${\bm X^\top \bm X}$ has the  expression
\begin{equation*}
{\bm X^\top \bm X}=\boldsymbol\Phi\boldsymbol\Lambda\boldsymbol\Phi^\top,
\end{equation*}
where $\boldsymbol\Lambda$ is a $p\times p$  diagonal matrix  with diagonal  entries being  the eigenvalues of $\bm X^\top\bm X$, $\boldsymbol\Phi=(\boldsymbol\phi_1,\boldsymbol\phi_2,...,\boldsymbol\phi_n)\in\R^{n\times n}$ is the eigenvector matrix of ${\bm X^\top \bm X}$ such that $\boldsymbol\Phi^\top\boldsymbol\Phi={\bm I}_p$. Then, with working number of factors $d$, the estimators of factor loadings and factors  are respectively
\begin{equation}\label{FL1}
\widetilde{\bm L}^d=\sqrt{p}(\boldsymbol\phi_1,\boldsymbol\phi_2,...,\boldsymbol\phi_d) \ \  \mbox{and} \ \  \widetilde{\bm F}^d=\frac{1}{{p}}{\bm X}\widetilde{\bm L}^d.
\end{equation}
Denote by $\text{Tr}(\bm A)$ the trace of matrix $\bm A$. It is easy to see that
\begin{equation*} \label{LF}
(\widetilde{\bm F}^{d},\widetilde{\bm L}^{d})=\underset{\substack{{\bm F}^d\in\R^{n\times d},{\bm L}^d\in\R^{p\times d}\\{{\bm L}^{d}}^\top {\bm L}^d/p={\bm I}_d}}{\text{argmin}}\{\text{Tr}[({\bm X}-{\bm F}^d{{\bm L}^{d}}^\top )({\bm X}-{\bm F}^d{{\bm L}^{d}}^\top )^\top ]\}.
\end{equation*}

\section{Double Cross Validation for Number of Factors}\label{DCVsection}

Our basic idea to specify the number of factor $d_0$, is based on the prediction error in a cross-validated manner. Write the model element-wise as
\begin{equation}\label{reg1}
{x}_{is}={{\bm f}_i^0}^\top {\bm l}^0_s+{{e}}_{is}, \ \ \ \ i=1, ..., n, \ \  s=1,2,...,p.
\end{equation}
CV needs to make prediction of $ {x}_{is} $ based on other elements except for itself.
Because neither $ {{\bm f}_i^0} $ nor $  {\bm l}^0_s $ in \eqref{reg1} is observable, we implement the prediction by two-stage fitting: leave-\emph{observation}-out  and leave-\emph{variable}-out.

The first stage is to estimate $  {\bm l}^0_s $ from the data by the $K$-fold CV. Divide the rows of $ {\bm X} $ into $K$ folds, $ 1< K \le n  $. Denote them by $M_1, ..., M_K$. Let $n_k=\#M_k$ be the number of elements in $M_k$, and ${\bm X}_{-M_k}$ be the sub-matrix of $ {\bm X} $ with  rows in $ M_k $ being removed. Apply  the eigen-decomposition to ${\bm X}_{-M_k}$ to obtain corresponding matrices $\widetilde{\bm F}^{k,d}$ and $\widetilde{\bm L}^{k,d}$ as in \eqref{FL1}.
Here, we use the superscript to highlight the fact that the estimator is obtained from the data with the rows in $M_k$ removed, with working number of factors $d$. Similar to \cite{BaiNg2002}, we rescale the estimated factor loading matrix and let
\begin{equation*}
\widehat{\bm L}^{k,d}= (\frac{1}{n-n_k}{\bm X}^\top_{-M_k}{{\bm X}_{-M_k}}) \widetilde{\bm L}^{k,d} = (\widehat{\bm l}^{k,d}_1, ..., \widehat{\bm l}^{k,d}_p).
\end{equation*}

In the second stage, we replace $ {\bm l}^0_s $ in \eqref{reg1} by $ \widehat{\bm l}^{k,d}_s $ and rewrite \eqref{reg1} as a regression model
\begin{equation}\label{reg2}
{x}_{is}={{\bm f}_i^0}^\top \widehat{\bm l}^{k,d}_s + {{e}}'_{is}, \ \ \ \ i \in M_k, \ \  s=1,2,...,p,
\end{equation}
where $ e'_{is} $ is the regression error in place of $ e_{is} $ due to the replacement. This time,  $\widehat{\bm l}^{k,d}_s, s=1,2,...,p$ are known and treated as the `regressors', but ${\bm f}^0_i$ is treated as  `regression coefficients' that need to be estimated.
By leaving $x_{is}$ out,  we estimate ${\bm f}^0_i$ by
\begin{equation*}
 \widehat {{\bm f}}_{i,s}^{k,d} = \underset{\bm f}{\text{argmin}} \sum_{t\neq s, t=1}^p ( {x}_{it} -\bm f^\top \widehat{\bm l}^{k,d}_t)^2 ,\  i \in M_k.
\end{equation*}
Thus we can predict $x_{is}$ by $(\widehat {{\bm f}}_{i,s}^{k,d})^\top \widehat{\bm l}^{k,d}_s $.
The average squared prediction error for $\bm x_i = (x_{i1}, ..., x_{ip})^\top$ is
\begin{equation*}
 V_i^{k,d} = \frac{1}{p}\sum\limits_{s=1}^p [x_{is} - (\widehat {{\bm f}}_{i,s}^{k,d})^\top \widehat{\bm l}^{k,d}_s]^2, \ \  i \in M_k.
\end{equation*}
The calculation of $  V_i^{k,d} $ can be much simplified  as shown in \cite{Shao1993}, i.e.,
\begin{equation*}
V_i^{k,d} =\frac{1}{p}\sum\limits_{s=1}^p(1-w^{k,d}_s)^{-2}(x_{is}-(\widehat{\bm f}_{i}^{k,d})^\top\widehat{\bm l}^{k,d}_s)^2,
\end{equation*}
where  $\widehat{\bm f}_i^{k,d}=(\widehat{\bm L}^{k,d\top}\widehat{\bm L}^{k,d})^{-1}\widehat{\bm L}^{k,d\top}\bm x_i$ is the conventional least squares estimator for \eqref{reg2} and $w^{k,d}_s$  is the $s$-th diagonal element of projection matrix $\Pro_{\widehat{\bm L}^{k,d}}=\widehat{\bm L}^{k,d}(\widehat{\bm L}^{k,d\top}\widehat{\bm L}^{k,d})^{-1} \times \widehat{\bm L}^{k,d\top}$.
Finally, consider the averaged  prediction error over all the elements
\begin{equation*}
DCV(d)=\frac{1}{n} \sum\limits_{k= 1}^K \sum\limits_{i \in M_k}^nV_i^{k,d}.
\end{equation*}
 Let $  d_{\max}$ be a fixed positive integer that is large enough such that  $  p> d_{\max} > d_0$.  The DCV estimator for the number of factors is given by
\begin{equation}\label{definek}
\widehat{d}_{DCV}=\underset{0\leq d\leq d_{\max}}{\text{argmin}}DCV(d).
\end{equation}

The above approach has similarity with the vanilla row-wise CV \citep{BKSK2008}, which also predicts ${\bm x}_{i}$  by two steps. Its first step is exactly the same as ours. However, in the second step, row-wise CV uses the full data instead of cross validation method,  leading to overfitting and inconsistency of the method. To fix this problem, \cite{EK1982}  suggested the element-wise cross validation. Although their methods solve the problem of overfitting, they are costly and involve immense computation.

Our DCV  offers a simpler solution. By utilizing CV in both the first stage of leaving ``observations'' out, and the second regression step of leaving   ``variables'' out, our method  guarantees that the prediction of each element does not use any information from itself, and thus the consistency is ensured as shown in the next section. In addition, the implement of $K$-fold CV can significantly reduce the computational complexity, especially when the number of rows is large. It is interesting to see that the $K$-fold CV in the first stage will also help in simplifying the calculation in the second stage. Because of this appealing nature, to further facilitate the calculation, we can transpose $ {\bm X} $ when $n<p$.


\section{Asymptotic consistency of the estimation}\label{sec:theory}
In this section, we investigate the consistency of our derived estimator.  We start with the following assumptions.
\begin{ass}\label{ass1}
The eigenvalues of ${\bm F}^{0\top}{\bm F}^0/n$ are bounded away from zero and infinity, and  that $\E\|{\bm f}^0_i\|^4<\infty$ for $ i=1, ..., n$.
\end{ass}
\begin{ass}\label{ass2} The eigenvalues of ${\bm L}^{0\top}{\bm L}^0/p$ are bounded away from zero and infinity, and that $\E\|{\bm l}^0_i\|^4<\infty$ for  $i=1, ..., p$.
 \end{ass}

\begin{ass}\label{ass3} There exists a constant $M<\infty$ such that
\begin{enumerate}
\item[1.]  $\E e_{it}=0$  and  $\E e_{it}^4\leq M$ for $1\leq i\leq n$ and $1\leq t\leq p$;

  \item[2.] $\gamma_p(i,j) = \E({\bm e}_i^\top{\bm e}_j/p)\leq M$  and $\frac{1}n\sum\limits_{j=1}^n\gamma^2_p(i,j)\leq M$ for all $i$;

 \item[3.] $\E({\bm e}_{is}{\bm e}_{it})=\tau_{st,i}$ with $\tau_{st,i}\leq |\tau_{st}|$ for some $\tau_{st}$ and for all $i$; in addition,  $\frac{1}p\sum\limits_{s=1}^p\sum\limits_{t=1}^p\tau_{st}$ $\leq M$;

 \item[4.] $\E({\bm e}_{is}{\bm e}_{jt})=\tau_{ij,st}$ and $\frac{1}{np}\sum\limits_{i=1}^n\sum\limits_{j=1}^n\sum\limits_{s=1}^p\sum\limits_{t=1}^p\tau_{ij,st}\leq M$;

\item[5.]  $\frac{1}{  n }\E|\sum\limits_{j=1}^n[e_{jt}e_{js}-\E(e_{jt}e_{js})]|^2\leq M$ for all $(s,t)$;

\item[6.]  $\|\frac{1}{\sqrt{n}}\sum\limits_{i=1}^n{\bm f}^{0}_{i}e_{is}\|^2\leq M$ for $1\leq s\leq p$;

\item[7.]  $\|\frac{1}{\sqrt{p}}\sum\limits_{s=1}^p{\bm l}^{0}_{s}e_{is}\|^2\leq M$ for $1\leq i\leq n$.

\end{enumerate}
\end{ass}

Assumptions \ref{ass1}--\ref{ass3} are commonly used in the  analysis of  factor models, for example, in \citet{BaiNg2002} and \citet{LLS2017}.  Assumptions \ref{ass1} and  \ref{ass2} together ensure that each factor plays a nontrivial role in contributing to the variation of ${\bm X}$. Unlike the strict factor model that assumes all entries of $\bm E$ are I.I.D., Assumption~\ref{ass3} allows weak dependency for elements of  $\bm E$, making our methods applicable to the approximate factor models.  These assumptions also indicate that  all the  eigenvalues of the covariance matrix of common components would dominate the eigenvalues corresponding to idiosyncratic components, which is crucial to the identifiability for the approximate factor models.

\begin{ass}\label{ass4}
$\sup\limits_k n_k=o(n^{1/3})$ and
 $K\sup\limits_kn_k/n< \infty$, where $n_k=\#M_k$ is the number of elements in $M_k$.

\end{ass}
Assumption \ref{ass4} is a weak condition on $K$-fold CV in the first stage. The number of elements in each fold can either be fixed or tend to infinity. In particular, this assumption includes the leave-one-out CV.
\edef\oldass{\the\numexpr\value{ass}+1}

\setcounter{ass}{0}
\renewcommand{\theass}{\oldass.\alph{ass}}

\begin{ass}\label{ass5a}  The idiosyncratic components satisfies
\begin{enumerate}
\item[1.] $\E[e_{it}^2\mid {\bm F}^0_{-M_k}{\bm L}^{0\top}]=\sigma^2$ for $i\in M_k$, where ${\bm F}^0_{-M_k}$ is a submatrix of ${\bm F}^0 $ with rows in $M_k$ being removed;

\item[2.] $\lambda_1({\bm E^\top}{\bm E}/n)<2 \sigma^2$.
\end{enumerate}

\end{ass}

\begin{ass}\label{ass5b} Let $p/n\to \rho\in(0,\infty)$.  There exists a $p\times p$ positive definite matrix $\boldsymbol\Sigma_p$ such that ${\bm e}_i=\boldsymbol\Sigma_p^{1/2}{\bm y}_i$, where ${\bm y}_i=( y_{i1},y_{i2},...,y_{ip})^\top$,
\begin{enumerate}
\item[1.] $\E y_{it}=0$, $\E y_{it}^4<\infty$,  and $y_{it}$ are I.I.D. for $1\leq i\leq n, 1\leq t\leq p$;

\item[2.]  ${\bm y}_i$ is independent of ${\bm F}^0_{-M_k}$;

 \item[3.] The spectral distribution of $\boldsymbol\Sigma_p$ is convergent, i.e., there exists a distribution function $H$ such that $\frac{1}{p}\sum\limits_{s=1}^p{\mathds{1}}(\lambda_s(\boldsymbol\Sigma_p)\leq t)\to H(t)$ as $p\to \infty$;

\item[4.] $2\text{\normalfont{diag}}(\boldsymbol\Sigma_p) - \boldsymbol\Sigma_p $ is  a positive definite matrix, where $\text{\normalfont{diag}}(\boldsymbol\Sigma_p)$ is a diagonal matrix that has the same diagonal elements as $\boldsymbol\Sigma_p$.
\end{enumerate}
\end{ass}
\let\theass\origtheass

Assumptions \ref{ass5a} and  \ref{ass5b} are two technical assumptions, either of which ensures the viability of the proposed method.  They also allow weak dependence among idiosyncratic components as well as the observations, which generalize the I.I.D. assumption required in the strict factor model.  In comparison, \cite{BaiNg2002} needs weaker assumptions on  the dependence than ours, possibly due to the relatively large penalty in their methods. Of course, this large penalty is at the cost of selecting inappropriately fewer factors as shown in our simulation study.




\begin{thm} \label{thm1}
Suppose  Assumptions \ref{ass1}--\ref{ass4} hold. If  either Assumption \ref{ass5a}  or \ref{ass5b} holds, then the DCV estimator \eqref{definek} is consistent as both $ n $ and $ p $ tend to infinity, i.e.,
\begin{equation*}
\lim\limits_{n,p\to\infty}Prob(\widehat{d}_{DCV}=d_0)=1.
\end{equation*}
\end{thm}

{\color{black}
\begin{rem}  Recent empirical findings suggest that the factor structure of economic data may change with  sample size and dimensionality of variable \citep{Onatski2010,  JLN2015,LLS2017}. In fact, Theorem \ref{thm1} can be easily extended to the case that $ d_0 $ varies slowly with both $ n $ and $ p $. {\color{red}DCV} is thus applicable to practical data with changing structure of factorization.
\end{rem}}

\section{Simulation Study}\label{sec:simulation}
In this section, we conduct simulation studies to compare the finite sample performances of the DCV with the other methods, including
the {panel} criterion $IC_1$ \citep{BaiNg2002} and $Ladle$ \citep{LL2016}.   
As shown in \cite{LL2016},  $Ladle$ estimator outperforms all the other eigen-structrue based methods in determining the number of factors, and is thus used  as a representative of those methods.
Similarly, for the IC-based methods, \cite{BaiNg2002} showed that their criterion outperforms other information criteria such as $AIC$ and $BIC$.
%
%
We do not consider other cross-validation methods as they are either inconsistent or suffer from excessive computational burden.
For our $DCV$, we consider both leave-one-out CV ($DCV_1$) and 10-fold CV ($DCV_{10}$) in the first stage.
Matlab codes for the calculations are available at \href{https://github.com/XianliZeng/Double-Cross-Validation}{https://github.com/XianliZeng/Double-Cross-Validation}.

\begin{figure*}[ht!]
\vspace{-0.3cm}
\centerline{\includegraphics[width=1.0\linewidth]{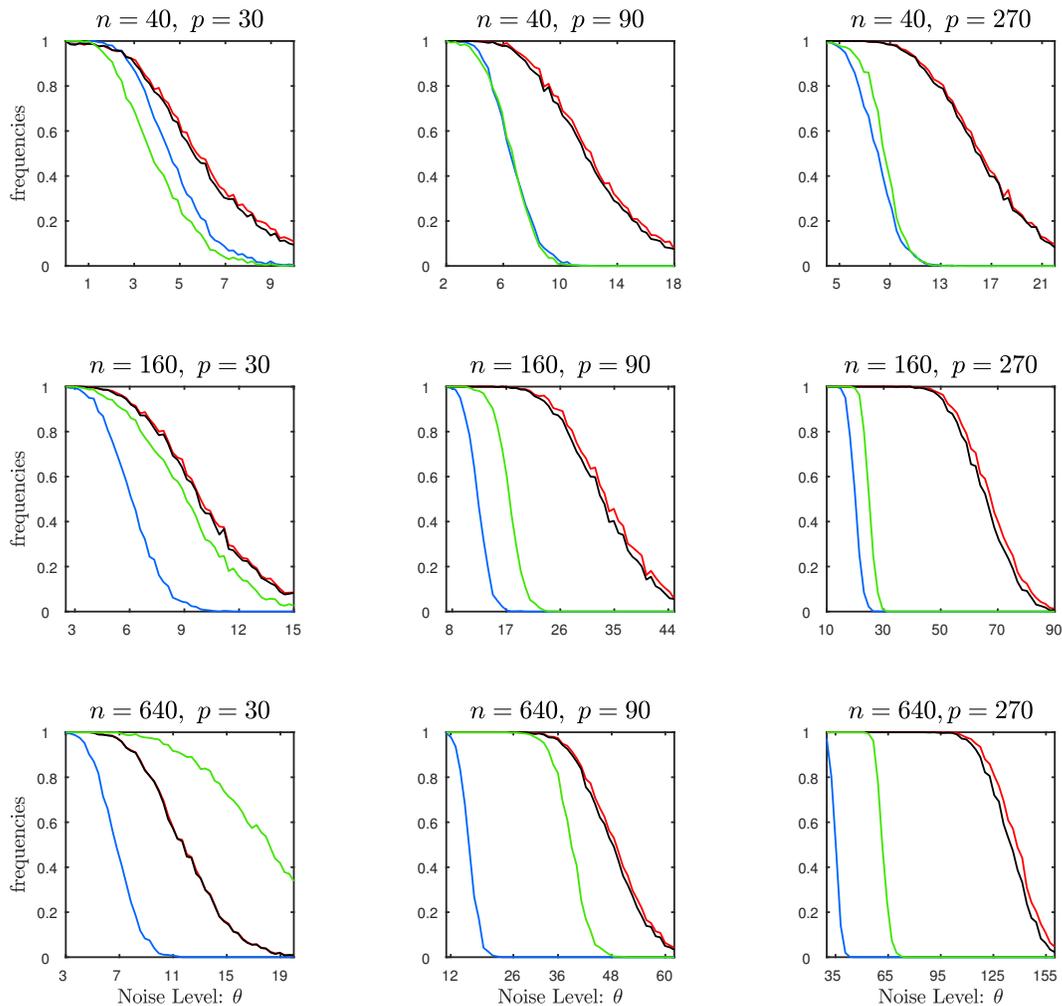} } \vspace{-0.5cm}
 \caption{\footnotesize The relative frequencies of correctly selecting the number of factors for model~\eqref{SM1} with independent Gaussian idiosyncratic errors (E1). In each panel, we use red line for $DCV$, black line for $DCV_{10}$, {blue} line for $IC_1$ and  green line for \textit{Ladle}.\vspace{-0.5cm}
}
\label{fig1}
\end{figure*}
\begin{figure*}[ht!]
\vspace{-0.3cm}
\centerline{\includegraphics[width=1.0\linewidth]{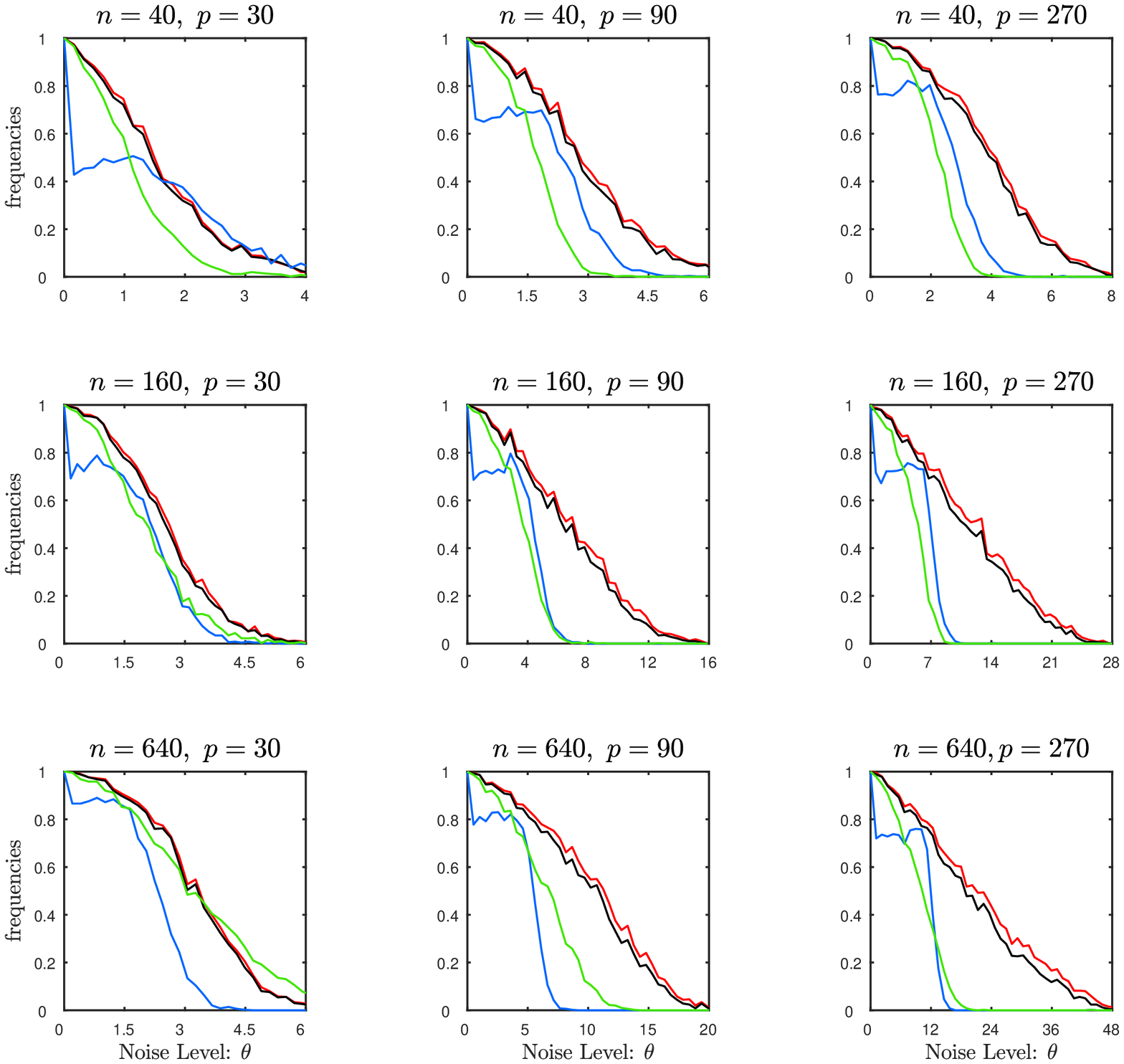} } \vspace{-0.5cm}
 \caption{\footnotesize The relative frequencies of correctly selecting the number of factors for  model \eqref{SM1} with independent $t$-distributed idiosyncratic errors (E2). In each panel, we use red line for $DCV$, black line for $DCV_{10}$, blue line for $IC_1$ and  green line for \textit{Ladle}.\vspace{-0.5cm}
}
\label{fig2}
\end{figure*}
\begin{figure*}[ht!]
\vspace{-0.3cm}
\centerline{\includegraphics[width=1.0\linewidth]{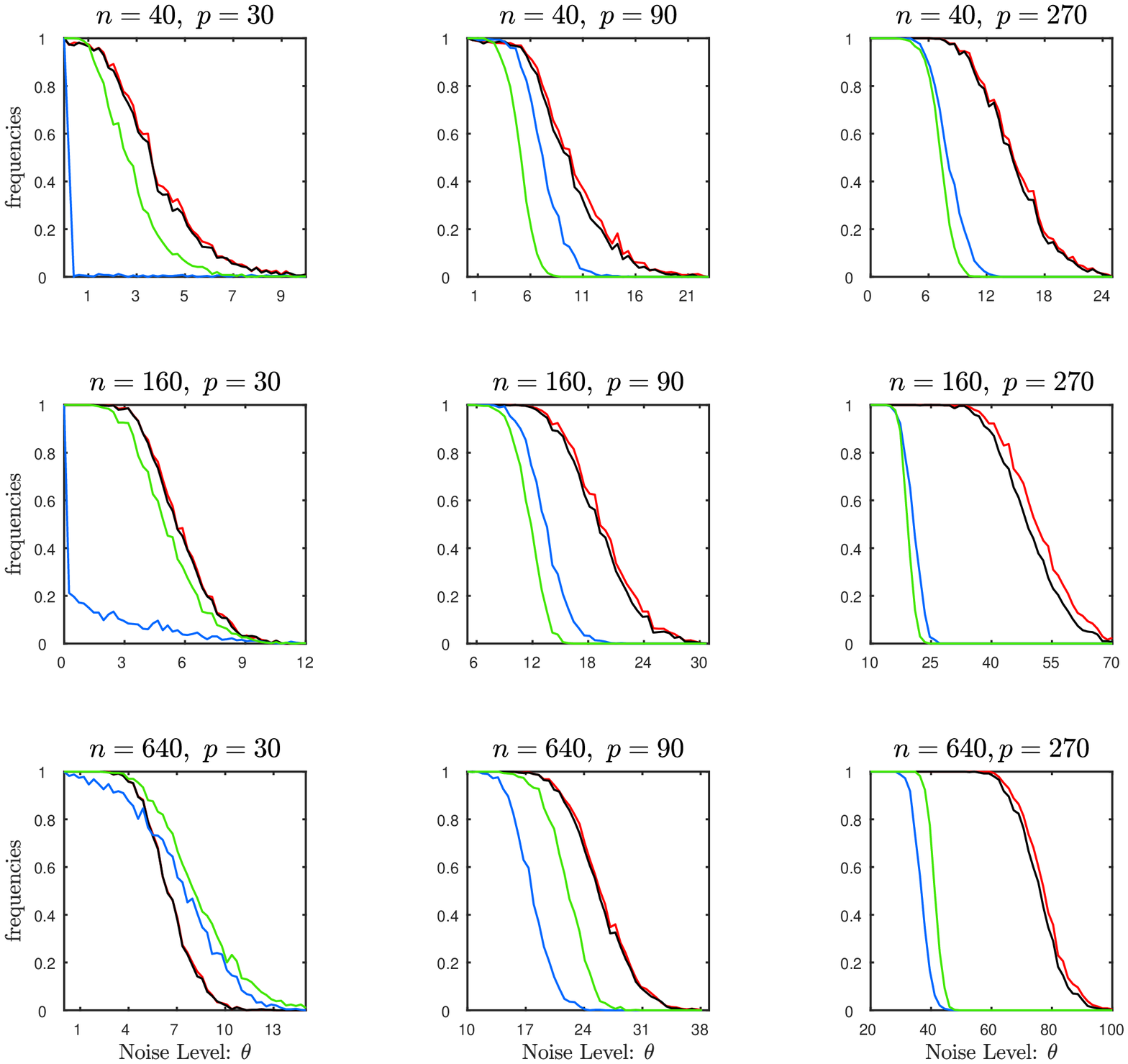} } \vspace{-0.5cm}
 \caption{\footnotesize The relative frequencies of correctly selecting the number of factors for   model~\eqref{SM1} with heteroskedastic idiosyncratic errors (E3). In each panel, we use red line for $DCV$, black line for $DCV_{10}$, blue line for $IC_1$ and  green line for \textit{Ladle}.\vspace{-0.5cm}
}
\label{fig3}
\end{figure*}

\begin{figure*}[ht!]
\vspace{-0.3cm}
\centerline{\includegraphics[width=1.0\linewidth]{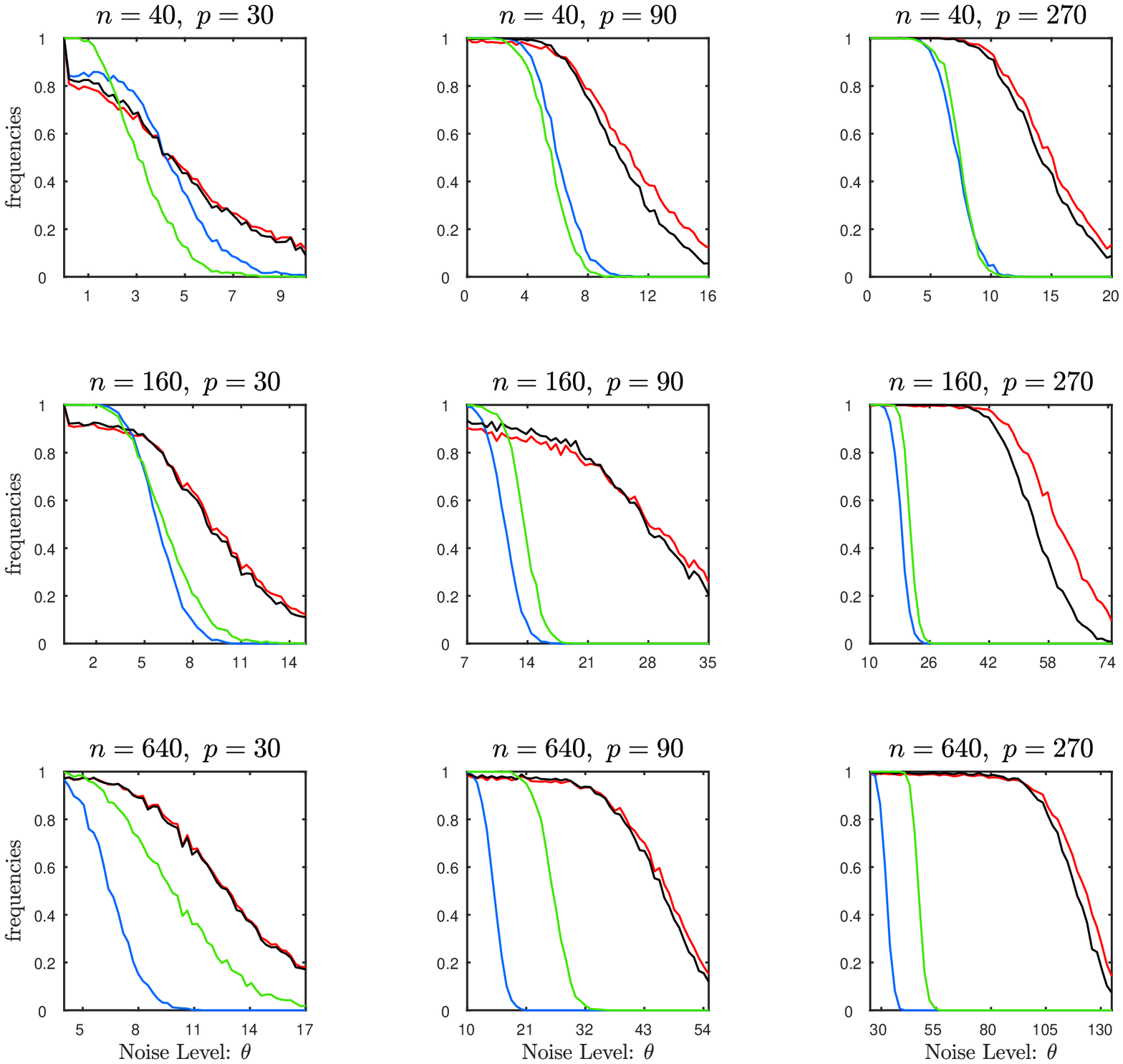} } \vspace{-0.5cm}
 \caption{\footnotesize The relative frequencies of correctly selecting the number of factors for  model~\eqref{SM1} with serial correlated  idiosyncratic errors (E4). In each panel, we use red line for $DCV$, black line for $DCV_{10}$, blue line for $IC_1$ and  green line for \textit{Ladle}.\vspace{-0.5cm}
}
\label{fig4}
\end{figure*}

\begin{figure*}[ht!]
\vspace{-0.3cm}
\centerline{\includegraphics[width=1.0\linewidth]{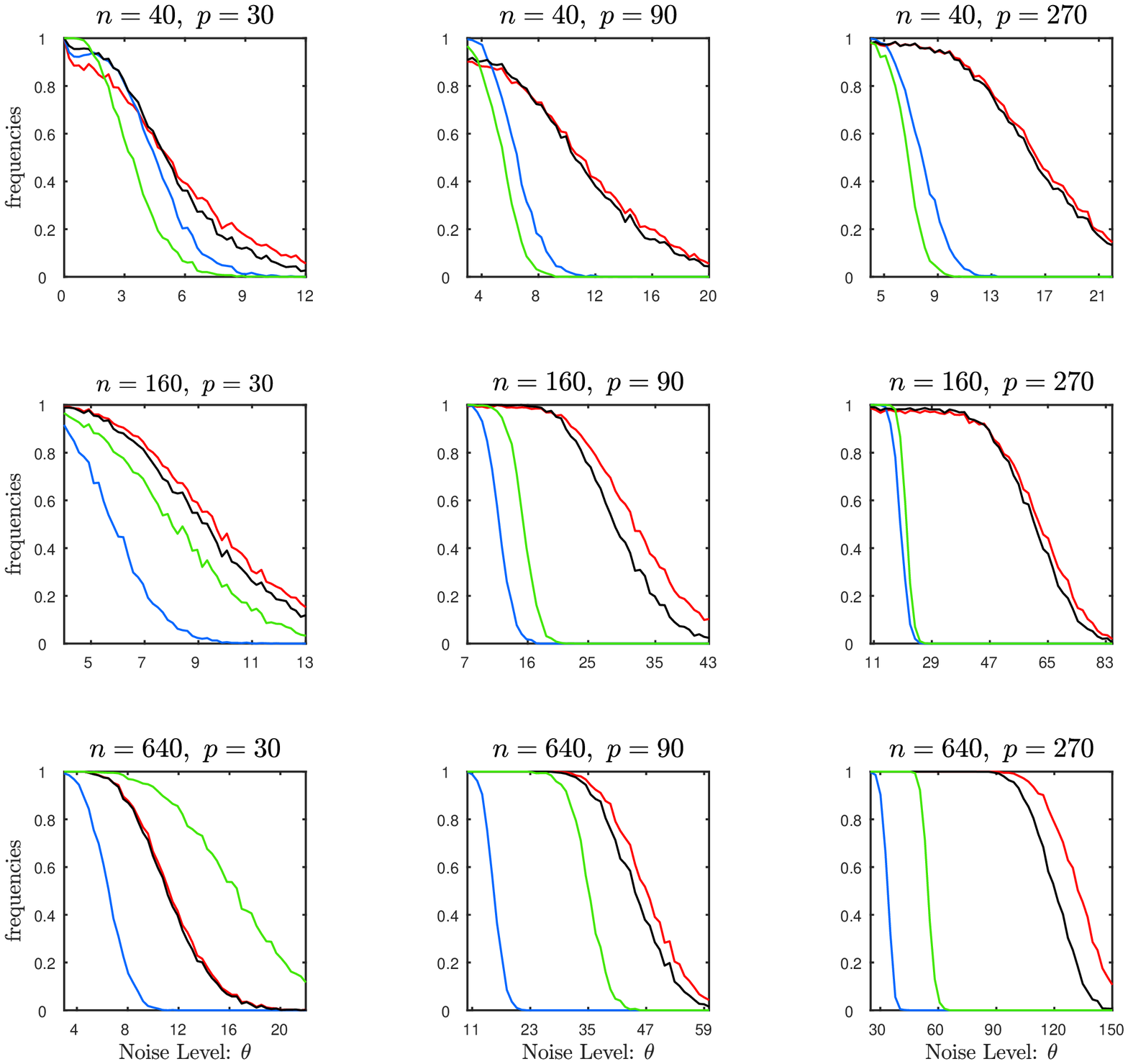} } \vspace{-0.5cm}
 \caption{\footnotesize The relative frequencies of correctly selecting the number of factors for  model~\eqref{SM1} with cross-sectional correlated idiosyncratic errors (E5). In each panel, we use red line for $DCV$, black line for $DCV_{10}$, blue line for $IC_1$ and  green line for \textit{Ladle}.\vspace{-0.5cm}
}
\label{fig5}
\end{figure*}

We simulate data from a factor model with $d_0 = 5$:
\begin{equation}
 x_{is}=\sum_{j=1}^{5}f_{ij}l_{sj}+\sqrt{\theta}e_{is}; \quad f_{ij},l_{sj}\sim N(0,1)  \ \ i=1, ..., n, \ \  s=1,2,...,p.  \label{SM1}
\end{equation}
For the idiosyncratic errors $ e_{is} $, the following settings are considered:
\begin{description}
\item[] (E1) independent Gaussian:  $ e_{is}\sim N(0,1)$;
\item[] (E2) independent t-distributed:  $ e_{is}\sim t_3 $;
\item[] (E3) heteroskedastic:  $ e_{is} \sim N(0,\delta_s), $ $ \delta_s = 1 $ if $s$ is odd,  and 2 if $s$ is even;
\item[] (E4) serial correlated: $ e_{is}=0.3 e_{i,s-1}+\nu_{is} $,  $\nu_{is}\sim N(0,1)$;
\item[] (E5) cross-sectional correlated: $  e_{is}=\sum\limits_{j=-10}^{10} 0.15^j \nu_{i-j,s}$,  $\nu_{is}\sim N(0,1)$.
\end{description}

In all simulations, we select $d$ from $\{1,2,...,d_{\max}\}$ where $d_{\max}= 8 $ as in \cite{BaiNg2002}.
We consider the combinations of { $n$ and $p$} with $n=\{40,160,640\}$ and $p=\{30,90, 270\}$.
To see the effect of  the signal-to-noise ratio on the methods, $\theta$ varies in a range of intervals depending on different settings of idiosyncratic errors and different $ n $ and $p$.
These ranges are made clear in each panel of the figures below. 
All comparisons  were evaluated based on 1000 replications.
Figures~\ref{fig1} to~\ref{fig5} summarize the simulation results, where the curves represent the frequencies of selecting the right number of factors at different noise level $\theta$.

Generally, the simulation results suggest that the performance of all the methods becomes better when $n$ and $p$ get larger, which lends support to the asymptotic consistency of the methods.
$Ladle$ usually shows better performance than $IC_1$, possibly due to its utilization of  both eigenvalues and eigenvectors.
Notably, $DCV_1$ and $DCV_{10}$ stand out and perform the best in most scenarios.
They possess a much higher frequency of correct selection when $\theta$ is large, implying their ability to select the right number of factors even when the {factors' contribution to the variation of the variables,  i.e., signal-to-noise ratio,} is low, which is particularly the situation in economic data.
Moreover, $DCV_1$ and $DCV_{10}$ present almost
identical performance for all the models, implying that the performance of the DCV is not sensitive to the choice of~$K$.

For different settings of the idiosyncratic errors, we have some additional observations.
Figure \ref{fig2} shows that $IC_1 $ is not so robust to the data with heavy-tail especially when the size of the data is small.
Figure \ref{fig3} implies that $IC_1 $ is also affected adversely by the heteroskedastic errors.
Note that both heavy tail and  heteroskedasticity are stylized facts in financial data.
In contrast, for these two types of data,  $DCV_1$ and $DCV_{10}$ still demonstrate very stable and efficient ability in selecting the number of factors.
In Figures \ref{fig4} and \ref{fig5}, where the idiosyncratic errors have serial or cross-sectional correlation,  DCV has inferior performance than $IC_1$ and $ Ladle $ when the variance of idiosyncratic components is small and  both $ n $ and $ p $ are small.
However, DCVs quickly gains advantage when either $ n $ or $ p $ increases.


\section{An empirical application: Fama-French  three-factor model}\label{sec:application}
\begin{figure*}[p]
\vspace{-0.5cm}
\centerline{Validation of three-factor model  based on value weighted return of 25 portfolios}
\centerline{\includegraphics[width=1\linewidth, height=0.56\linewidth]{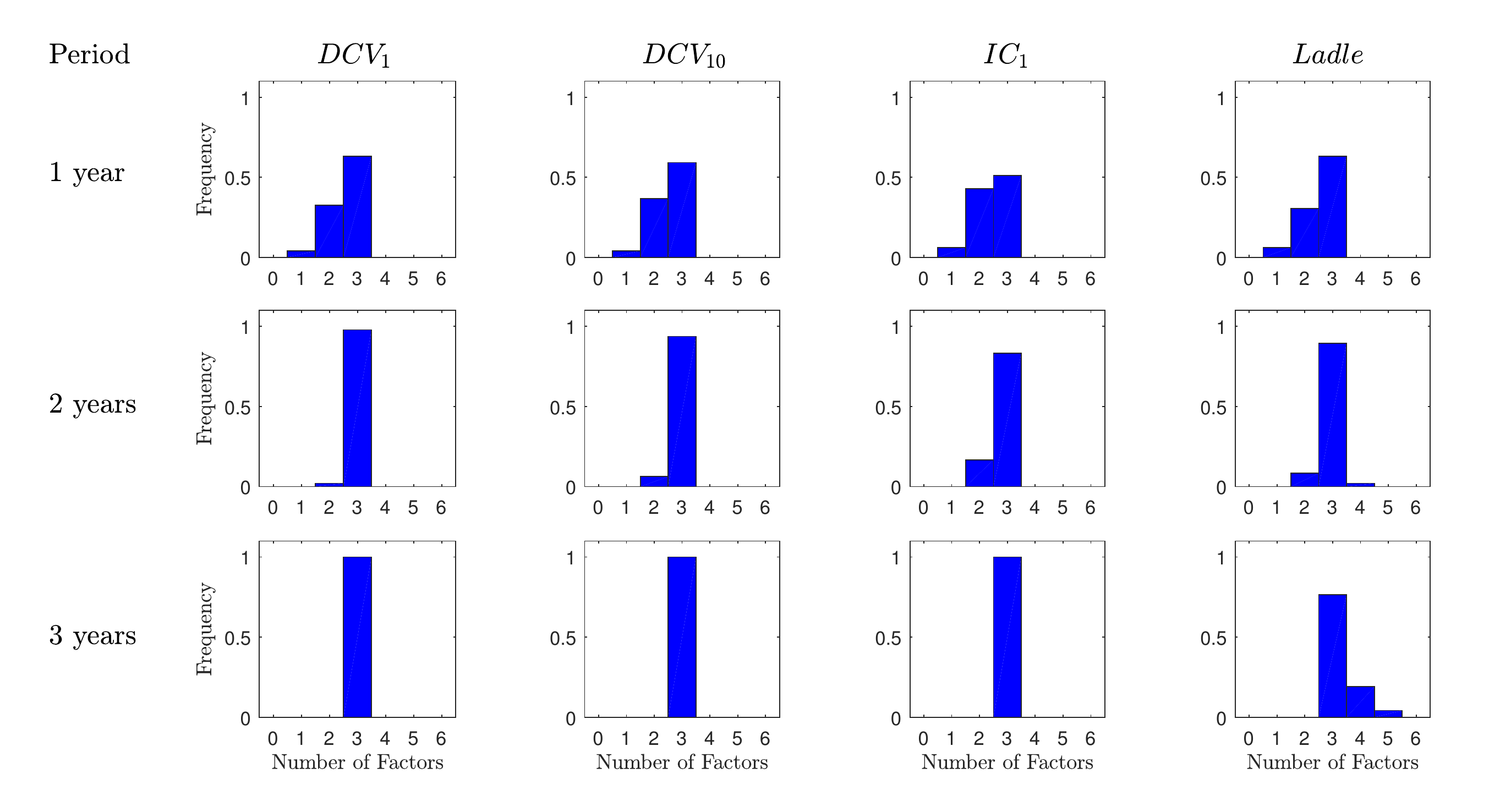}}

\centerline{Validation of three-factor model  based on value weighted return of 100 portfolios}
\centerline{\includegraphics[width=1\linewidth, height=0.56\linewidth]{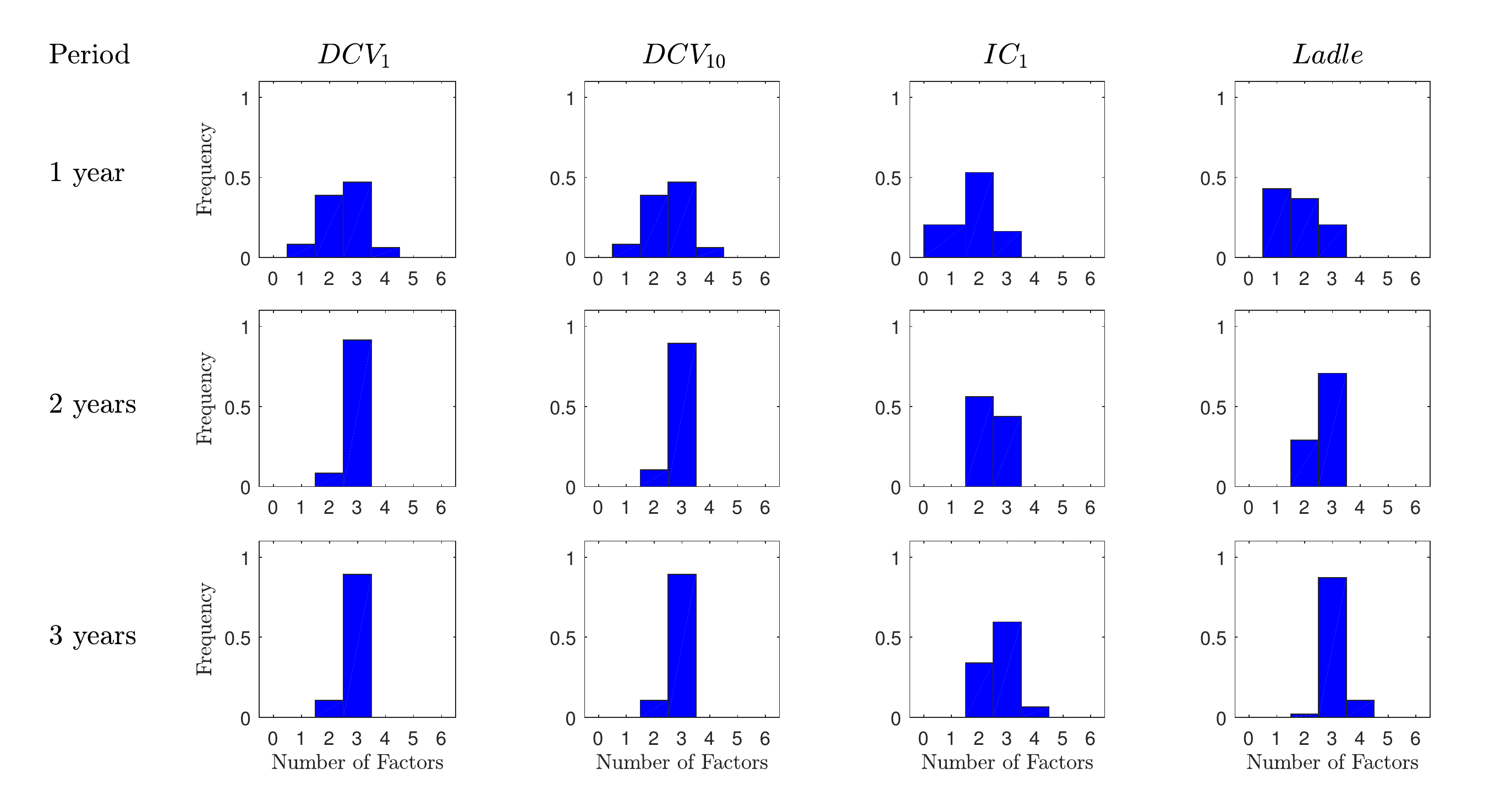}} \vspace{-0.5cm}
 \caption{\footnotesize The  frequencies of selected numbers of factors for the value weighted returns of 25 and 100 portfolios. Panels in each column represents one method; panels in the same rows are based on the same data in fixed period of years.
}
\label{fig6}
\end{figure*}

In asset pricing and portfolio management, the so-called three-factor model of \cite{FF1993} is widely applied to describe the stock returns.
It shows that the excess return of a stock or a portfolio ($R_{it}-{r_{ft}}$) can be satisfactorily explained by three common factors:  (1) excess return of the market portfolio ($R_{Mt}-r_{ft}$), (2) the outperformance of small-cap companies versus big-cap companies ($SMB$), and (3) the outperformance of companies with  high book to market ratio versus those with small book to market ratio ($HML$), i.e.,
\begin{equation*}
R_{it}-r_{ft}=\beta_1(R_{Mt}-r_{ft})+\beta_2SMB_t+\beta_{3}HML_t+\epsilon_{it},
\end{equation*}
where $ r_{ft} $ is the risk-free return.
Next,  we use the data provided by Professor  Kenneth R. French\footnote{http://mba.tuck.dartmouth.edu/pages/faculty/ken.french/data\_library.html} to mutually validate our methods and the three-factor model.
%
Stocks are divided into 6, 25 or 100 groups according to company's capitalization (market equity) and  Book-to-Market ratio, and portfolios are constructed with  value weighted daily returns or equally weighted daily returns in each group.
As those portfolios were compiled based on company's data at the end of June of every year, we set the first of July of every year as a starting date of a period.
It is understandable that the  portfolios change from year to year, thus we only consider periods of 1 to 3 years in our calculation.

Figure \ref{fig6} presents the frequencies of the numbers of factors detected by the $DCV_1$, $DCV_{10}$, $IC_1$ and $Ladle$ in two of the data sets with $d_{\max}$ set at $15$.
In most situations, all the  approaches identify three common factors for the returns of portfolios, which are in line with the \emph{three}-factor model of \cite{FF1993}.
In general,  DCV selects three factors most often amongst the approaches, showing its superior performance in practice.

\section{Concluding Remarks}\label{sec:discussion}
In the literature, both CV and IC based approaches were popularly used for determining the complexity of a model, such as  the number of factors.
As noticed in \cite{BKSK2008}, most CV approaches for factor models are not consistent due to their insufficient validation, while existing remedies for the consistency is very computationally expensive.
In contrast, ICs are consistent and easy to implement, but they depend on  predetermined penalty  functions \citep{BaiNg2002}.
Simulation results not reported here show that the penalty function in \cite{BaiNg2002} could  be modified to be more adaptive to the data.
In addition, ICs are unstable when the signal-to-noise ratio is relatively low \citep{Onatski2010}.

By validating the model twice, the proposed DCV method does not only ensure the consistency under mild conditions, but is also easy to implement. %
Because DCV is based on prediction error, it automatically selects the number of factors that balances the model complexity and stability.
Our simulation studies also demonstrate its superior efficiency over the existing methods most of the time.
The only exception is when there is strong serial dependence in the idiosyncratic errors.
However, this deficiency disappears when the sample size and dimension of the variables increase.
The advantages of our method are more  pronounced for data with relative large variation, heteroscedasticity  or heavy tails in the idiosyncratic errors. The method is thus particularly relevant for financial data.
%


\bibliographystyle{apalike}
\bibliography{paper}


\

\appendix

\section{Proofs for Theorems}\label{sec:supplement}
In this section, we provide detailed proofs for our theoretical results, mainly  based on the advanced techniques derived in  \cite{BaiNg2002} and RMT \citep{BZ2008, LP2011,BPZ2015}. We first introduce several notations that repeatedly appear hereafter. For matrix ${\bm A}=(a_{ij})^{n\times p}$, ${\bm a}_i = (a_{i1}, a_{i2}, \cdots , a_{ip})^\top $ is the transpose of the $i$-th row,  and ${\bm A}_{-M_k}$ is the sub-matrix of ${\bm A}$ with rows in $M_i$ being removed,  and $\lambda_{t}({\bm A})$ is the $t$-th largest eigenvalue of $\bm A$, and $\|\bm A\|$ is the  Frobenius norm of $\bm A$.
 Let  $\widehat{\lambda}_t, t=1,2,...,p$ be the eigenvalues of ${\bm X}^\top{\bm X}$  with decreasing order  and $\widehat{\boldsymbol\phi}_{t}, t=1,2,...,p$ be the corresponding eigenvectors. Let $({\lambda}_t,\boldsymbol\phi_t,t=1,...,p)$, $(\widehat{\lambda}^k_t,\widehat{\boldsymbol\phi}^k_t,t=1,...,p)$ and $({\lambda}^k_t,\boldsymbol\phi^k_t,t=1,...,p)$ be corresponding quantities for  ${\bm L}^0{\bm F}^{0\top}{\bm F}^0{\bm L}^{0^\top}$, ${\bm X}^\top_{-M_k}{\bm X}_{-M_k}$ and
${\bm L}^0{\bm F}^{0\top}_{-M_k}{\bm F}^0_{-M_k}{\bm L}^{0^\top}$, respectively. We use $M_{np}$ to denote $\max(n,p)$ and $m_{np}$ to denote $\min(n,p)$. In addition, we use  $C$and $C'$ to denote  constants that may vary in different places throughout the paper.

To make the proofs easy to read, we give the proof to our main result in the first subsection, and put supporting facts and their proofs in the second subsection of this Appendix.

\subsection{Proof of Theorem \ref{thm1}}
Theorem \ref{thm1} follows if we can show that
\begin{equation*}
\lim_{n,p\to\infty}\text{Prob}(DCV({d})<DCV(d_0))=0 \text{ for } d\neq d_0 \text{ and } 1\leq d\leq d_{\max}.
\end{equation*}
Let $w^{k,d}_s$ be  the $s$th diagonal element of $\Pro_{\widehat{\bm L}^{k,d}}=\widehat{\bm L}^{k,d}(\widehat{\bm L}^{k,d\top}\widehat{\bm L}^{k,d})^{-1}\widehat{\bm L}^{k,d\top}$, $W_i(d,{\bm L})=\frac{1}{p}{\bm x}_i^\top({\bm I}_d-\Pro_{{\bm L}}){\bm x}_i$ and ${\bm H}^{k,d}=\frac{1}{(n-n_k)p}{\bm F}^{0\top}_{-M_k}{\bm X}^\top_{-M_k}\widetilde{\bm L}^{k,d}$. For $1\leq d< d_0$, according to \cite{Shao1993},
\begin{eqnarray*}
DCV(d)&=&\frac{1}{np}\sum\limits_{k=1}^K\sum\limits_{i\in M_k}\sum\limits_{s=1}^p(1-w^{k,d}_s)^{-2}(x_{is}-\widehat{\bm f}_i^{k,d\top}\widehat{\bm l}^{k,d}_s)^2\\
&=&\frac{1}{np}\sum\limits_{k=1}^K\sum\limits_{i\in M_k}\sum\limits_{s=1}^p(1+2w^{k,d}_s+O({w^{k,d}_s}^2))(x_{is}-\widehat{\bm f}_i^{k,d\top}\widehat{\bm l}^{k,d}_s)^2\\&=&\frac{1}{np}\sum\limits_{k=1}^K\sum\limits_{i\in M_k}{\bm x}_i^\top({\bm I}_d-\Pro_{\widehat{\bm L}^{k,d}}){\bm x}_i+o_p(1)\\
&=&\frac{1}{n}\sum\limits_{k=1}^K\sum\limits_{i\in M_k}W_i(d,\widehat{\bm L}^{k,d})+o_p(1).
\end{eqnarray*}
Then,
\begin{eqnarray*}
DCV(d)-DCV(d_0)&=&\frac{1}{n}\sum\limits_{k=1}^K\sum\limits_{i\in M_k}W_i(d,\widehat{\bm L}^{k,d})-\frac{1}{n}\sum\limits_{k=1}^K\sum\limits_{i\in M_k}W_i(d_0,\widehat{\bm L}^{k,d_0})+o_p(1)\\
&=&\frac{1}{n}\sum\limits_{k=1}^K\sum\limits_{i\in M_k}W_i(d,\widehat{\bm L}^{k,d})-\frac{1}{n}\sum\limits_{k=1}^K\sum\limits_{i\in M_k}W_i(d,{\bm L}^0{\bm H}^{k,d})\\
&&+\frac{1}{n}\sum\limits_{k=1}^K\sum\limits_{i\in M_k}W_i(d,{\bm L}^0{\bm H}^{k,d})-\frac{1}{n}\sum\limits_{k=1}^K\sum\limits_{i\in M_k}W_i(d_0,{\bm L}^0)\\
&&+\frac{1}{n}\sum\limits_{k=1}^K\sum\limits_{i\in M_k}W_i(d_0,{\bm L}^0)-\frac{1}{n}\sum\limits_{k=1}^K\sum\limits_{i\in M_k}W_i(d_0,{\bm L}^0{\bm H}^{k,d_0})\\
&&+\frac{1}{n}\sum\limits_{k=1}^K\sum\limits_{i\in M_k}W_i(d_0,{\bm L}^0{\bm H}^{k,d_0})-\frac{1}{n}\sum\limits_{k=1}^K\sum\limits_{i\in M_k}W_i(d_0,\widehat{\bm L}^{k,d_0})+o_p(1)\\
&>&c_d+o_p(1).
\end{eqnarray*}
Here, the last inequality follows from Lemma \ref{lem3}, Lemma \ref{lem4} and the fact that $\Pro_{{\bm L}^0}=\Pro_{{\bm L}^0{\bm H}^{k,d_0}}$. Consequently,
\begin{equation*}
\lim_{n,p\to\infty}\text{Prob}(DCV({d})<DCV(d_0))=0.
\end{equation*}
Let $\mathcal{Q}_{\widehat{\bm L}^{k,d}}=\text{diag}(w^{k,d}_1,w^{k,d}_2,...,w^{k,d}_p)$ and $\text{Cov}({\bm e}_i)=\boldsymbol\Sigma_{p,i}$. For $d_0<d\leq d_{\max}$,  according to Lemma \ref{lem7},
\begin{eqnarray*}
&&DCV(d)\\
&=&\frac{1}{np}\sum\limits_{k=1}^K\sum\limits_{i\in M_k}\sum\limits_{s=1}^p(1-w^{k,d}_{s})^{-2}(x_{is}-\widehat{\bm f}_i^{k,d\top}\widehat{\bm l}^{k,d}_s)^2\\
&=&\frac{1}{np}\sum\limits_{k=1}^K\sum\limits_{i\in M_k}\sum\limits_{s=1}^p(1+2w^{k,d}_{s}+O({w^{k,d}_{s}}^2))(x_{is}-\widehat{\bm f}_i^{k,d\top}\widehat{\bm l}^{k,d}_s)^2\\
&=&\frac{1}{np}\sum\limits_{k=1}^K\sum\limits_{i\in M_k}{\bm x}_i^\top({\bm I}_d-\Pro_{\widehat{\bm L}^{k,d}}){\bm x}_i+\frac{2}{np}\sum\limits_{k=1}^K\sum\limits_{i\in M_k}\sum\limits_{s=1}^pw^{k,d}_{s}(x_{is}-\widehat{\bm f}_i^{k,d\top}\widehat{\bm l}^{k,d}_s)^2+o_p(\frac{1}{p})\\
&=&\frac{1}{np}\sum\limits_{k=1}^K\sum\limits_{i\in M_k}{\bm x}_i^\top({\bm I}_d-\Pro_{\widehat{\bm L}^{k,d}}){\bm x}_i+\frac{2}{np}\sum\limits_{k=1}^K\sum\limits_{i\in M_k} \text{Tr}(\mathcal{Q}_{\widehat{\bm L}^{k,d}}\boldsymbol\Sigma_{p,i})+o_p(\frac{1}{p}).\\
\end{eqnarray*}
Therefore,
\begin{eqnarray*}
&&DCV(d)-DCV(d_0)\\
&=&\frac{1}{np} \sum\limits_{k=1}^K\sum\limits_{i\in M_k}{\bm x}_i^\top(\Pro_{\widehat{\bm L}^{k,d_0}}-\Pro_{\widehat{\bm L}^{k,d}}){\bm x}_i+\frac{2}{np}\sum\limits_{k=1}^K\sum\limits_{i\in M_k} \text{Tr}((\mathcal{Q}_{\widehat{\bm L}^{k,d}}-\mathcal{Q}_{\widehat{\bm L}^{k,d_0}})\boldsymbol\Sigma_{p,i})+o_p(\frac{1}{p})\\
&=&-\frac{1}{np}\sum_{l=d_0+1}^d \sum\limits_{k=1}^K\sum\limits_{i\in M_k}\|{\bm x}_i^\top\widehat{\boldsymbol\phi}^k_l\|^2+\frac{2}{np}\sum\limits_{k=1}^K\sum\limits_{i\in M_k} \text{Tr}((\mathcal{Q}_{\widehat{\bm L}^{k,d}}-\mathcal{Q}_{\widehat{\bm L}^{k,d_0}})\boldsymbol\Sigma_{p,i})+o_p(\frac{1}{p}).
\end{eqnarray*}
According to Lemma \ref{lem14},
\begin{equation*}
\frac{1}{n}\sum_{l=d_0+1}^d\sum\limits_{k=1}^K\sum\limits_{i\in M_k}\|{\bm x}_i^\top\widehat{\boldsymbol\phi}^k_l\|^2< \frac{2}{n}\sum\limits_{k=1}^K\sum\limits_{i\in M_k} \text{Tr}((\mathcal{Q}_{\widehat{\bm L}^{k,d}}-\mathcal{Q}_{\widehat{\bm L}^{k,d_0}})\boldsymbol\Sigma_{p,i})+o_p(1).
\end{equation*}
It follows that,
\begin{equation*}
\lim_{n,p\to\infty}\text{Prob}(DCV(d)<DCV(d_0))=0.
\end{equation*}
We thus complete the proof of Theorem  \ref{thm1}.

\

\subsection{Supporting Lemmas and their Proofs}\label{sec73}

\begin{lem}\label{lem1}
Suppose  Assumptions \ref{ass1}--\ref{ass4} hold. Let $\widehat{\bm l}_{s}^{k,d}$ and ${\bm l}_s^0$ be the transpose of the $s$-th row of $\widehat{\bm L}^{k,d}$ and ${\bm L}^0$ respectovely. For any $k$ and $d$,  there exist a $(d_0\times d)$ matrix ${\bm H}^{k,d}$ with \normalfont{rank}$({\bm H}^{k,d})=\min(d_0,d)$ such that,
\begin{equation*}
\|\widehat{\bm l}_{s}^{k,d}-{\bm H}^{k,d\top}{\bm l}_s^0\|^2=O_p(m_{np}^{-1}),      \ \ \ \ \text{ for each } k \text { and }  s.
\end{equation*}
\end{lem}
\begin{proof}
Let ${\bm H}^{k,d}=\frac{1}{(n-n_k)p}{\bm F}^{0\top}_{-M_k}{\bm X}^\top_{-M_k}\widetilde{\bm L}^{k,d},$ then,
 \begin{eqnarray*}
\widehat{\bm L}^{k,d}-{\bm L}^0{\bm H}^{k,d}&=&\frac{1}{(n-n_k)p}({\bm X}^\top_{-M_k}-{\bm L}^{0}{\bm F}^{0\top}_{-M_k}){\bm X}_{-M_k}\widetilde{\bm L}^{k,d}\\
&=&\frac{1}{(n-n_k)p}{\bm E}^\top_{-M_k}{\bm X}_{-M_k}\widetilde{\bm L}^{k,d}\\
&=&\frac{1}{(n-n_k)p}{\bm E}^\top_{-M_k}{\bm E}_{-M_k}\widetilde{\bm L}^{k,d}+\frac{1}{(n-n_k)p}{\bm E}^\top_{-M_k}{\bm F}_{-M_k}^{0}{\bm L}^{0\top}\widetilde{\bm L}^{k,d}.
\end{eqnarray*}
Thus,
\begin{eqnarray*}
 \|\widehat{\bm l}_{s}^{k,d}-{\bm H}^{k,d\top}{\bm l}_s^0\|^2&=&\left\|\frac{1}{p}\sum\limits_{t=1}^p\widetilde{\bm l}_t^{k,d}\left[\frac{1}{n-n_k}\sum\limits_{i\notin M_k}e_{it}e_{is}\right]+\frac{1}{p}\sum\limits_{t=1}^p\widetilde{\bm l}_t^{k,d}{\bm l}_t^{0\top}\frac{1}{n-n_k}\sum\limits_{i\notin M_k}{\bm f}^{0}_{i}e_{is}\right\|^2\\
&\leq&3\left\|\frac{1}{p}\sum\limits_{t=1}^p\widetilde{\bm l}_t^{k,d}\left(\frac{1}{n-n_k}\sum\limits_{i\notin M_k}[e_{it}e_{is}-\E(e_{it}e_{is})]\right)\right\|^2\\
&&+3\left\|\frac{1}{p}\sum\limits_{t=1}^p\widetilde{\bm l}_t^{k,d}\frac{1}{n-n_k}\sum\limits_{i\notin M_k}\E(e_{it}e_{is})\right\|^2\\
&&+3\left\|\frac{1}{p}\sum\limits_{t=1}^p\widetilde{\bm l}_t^{k,d}{\bm l}_t^{0\top}\frac{1}{n-n_k}\sum\limits_{i\notin M_k}{\bm f}^{0}_{i}e_{is}\right\|^2\\
&=&3[a(k,t)+b(k,t)+c(k,t)].
\end{eqnarray*}

Next, we investigate the three terms separately as follows.
\begin{eqnarray*}
a(k,t)&\leq &\frac{1}{p^2}\sum\limits_{t=1}^p\|\widetilde{\bm l}_t^{k,d}\|^2\sum\limits_{t=1}^p\left(\frac{1}{n-n_k}\sum\limits_{i\notin M_k}[e_{it}e_{is}-\E(e_{it}e_{is})]\right)^2\\
&=&\frac{d}{p}\sum\limits_{t=1}^p\left(\frac{1}{n-n_k}\sum\limits_{i\notin M_k}[e_{it}e_{is}-\E(e_{it}e_{is})]\right)^2\\
&=&O_p(n^{-1}).
\end{eqnarray*}
\begin{eqnarray*}
b(k,t)&\leq &\frac{1}{p^2}\sum\limits_{t=1}^p\|\widetilde{\bm l}_t^{k,d}\|^2\sum\limits_{t=1}^p\left(\E\bigg(\frac{1}{n-n_k}\sum\limits_{i\notin M_k}e_{it}e_{is}\bigg)\right)^2\\
&=&\frac{d}{p}\sum\limits_{t=1}^p\left(\E\bigg(\frac{1}{n-n_k}\sum\limits_{i\notin M_k}e_{it}e_{is}\bigg)\right)^2\\
&=&O_p(p^{-1}).
\end{eqnarray*}
\begin{eqnarray*}
c(k,t)&\leq &\frac{1}{p^2}\left\|\sum\limits_{t=1}^p\widetilde{\bm l}_t^{k,d}{\bm l}_t^{0\top}\frac{1}{n-n_k}\sum\limits_{i\notin M_k}{\bm f}^{0}_{i}e_{is}\right\|^2\\
&\leq&\frac{1}{n-n_k}\left\|\frac{1}{\sqrt{n-n_k}}\sum\limits_{i\notin M_k}{\bm f}^{0}_{i}e_{is}\right\|^2\frac{1}{p}\sum\limits_{t=1}^p\|\widetilde{\bm l}_t^{k,d}\|^2\frac{1}{p}\sum\limits_{t=1}^p\|{\bm l}_t^{0}\|^2\\
&=&O_p(n^{-1}).
\end{eqnarray*}
In summary,
\begin{equation*}
\|\widehat{\bm l}_{s}^{k,d}-{\bm H}^{k,d\top}{\bm l}_s^0\|^2=O_p(m_{np}^{-1}).
\end{equation*}
\end{proof}

\begin{lem}\label{lem2}
Suppose  Assumptions \ref{ass1}--\ref{ass4} hold. For any matrix $\bm F$, let $\Pro_{\bm F}={\bm F}({\bm F}^\top{\bm F})^{-1}{\bm F}^\top$ be its projection matrix. If the random vector ${\bm v}$ satisfies $\|{\bm v}\|^2=O_p(1)$,
\begin{equation*}{\bm v}^\top(\Pro_{{\bm L}^0{\bm H}^{k,d}}-\Pro_{\widehat{\bm L}^{k,d}}){\bm v}=O_p(m_{np}^{-1/2}).
\end{equation*}
\end{lem}
\begin{proof}
 Let ${\bm D}_{k,0}={\bm H}^{k,d\top}{\bm L}^{0\top}{\bm L}^0{\bm H}^{k,d}/p$ and ${\bm D}_{k,d}=(\widehat{\bm L}^{k,d})^\top\widehat{\bm L}^{k,d}/p$. We have $\|{\bm D}_{k,d}^{-1}\|=O(1)$ and $\|{\bm D}_{k,0}^{-1}-{\bm D}_{k,d}^{-1}\|=O_p(m_{np}^{-1/2})$ (see Lemma 2 of \cite{BaiNg2002}).
Note that,
\begin{eqnarray*}
&&{\bm v}^\top(\Pro_{{\bm L}^0{\bm H}^{k,d}}-\Pro_{\widehat{\bm L}^{k,d}}){\bm v}\\
&=&\frac{1}{p}{\bm v}^\top{\bm L}^0{\bm H}^{k,d}({\bm D}_{k,0}^{-1}-{\bm D}_{k,d}^{-1}){\bm H}^{k,d\top}{\bm L}^{0\top}{\bm v}\\
&&-\frac{1}{p}{\bm v}^\top(\widehat{\bm L}^{k,d}-{\bm L}^0{\bm H}^{k,d}){\bm D}_{k,d}^{-1}(\widehat{\bm L}^{k,d}-{\bm L}^0{\bm H}^{k,d})^\top{\bm v}\\
&&-\frac{1}{p}{\bm v}^\top(\widehat{\bm L}^{k,d}-{\bm L}^0{\bm H}^{k,d}){\bm D}_{k,d}^{-1}{\bm H}^{k,d\top}{\bm L}^{0\top}{\bm v}\\
&&-\frac{1}{p}{\bm v}^\top{\bm L}^{0}{\bm H}^{k,d}{\bm D}_{k,d}^{-1}(\widehat{\bm L}^{k,d}-{\bm L}^0{\bm H}^{k,d})^\top{\bm v}\\
&=&I+II+III+IV.
\end{eqnarray*}
We consider the four terms in turn.
\begin{eqnarray*}
 I&=&\frac{1}{p}\sum\limits_{s=1}^p\sum\limits_{t=1}^pv_{s}{\bm l}_s^{0\top}{\bm H}^{k,d}({\bm D}_{k,0}^{-1}-{\bm D}_{k,d}^{-1}){\bm H}^{k,d\top}{\bm l}_t^0v_{t}\\.
 &\leq&\frac{1}{p}\|{\bm D}_{k,0}^{-1}-{\bm D}_{k,d}^{-1}\|(\|\sum\limits_{s=1}^pv_{s}{\bm l}_s^{0\top}{\bm H}^{k,d}\|^2)\\
 &\leq&\frac{1}{p}\|{\bm D}_{k,0}^{-1}-{\bm D}_{k,d}^{-1}\|\cdot\sum\limits_{s=1}^p\|{\bm l}_s^{0\top}{\bm H}^{k,d}\|^2\cdot \|{\bm v}\|^2\\
  &=&O_p({m_{np}^{-1/2}});
\end{eqnarray*}
\begin{eqnarray*}
 II&=&\frac{1}{p}\sum\limits_{s=1}^p\sum\limits_{t=1}^pv_{s}(\widehat{\bm l}^{k,d}_s-{\bm H}^{k,d\top}{\bm l}_s^0)^\top{\bm D}_{k,d}^{-1}(\widehat{\bm l}^{k,d}_t-{\bm H}^{k,d\top}{\bm l}_t^0)v_{t}\\
 &\leq&\frac{1}{p}\Big\{\sum\limits_{s=1}^p\sum\limits_{t=1}^p\big[(\widehat{\bm l}^{k,d}_s-{\bm H}^{k,d\top}{\bm l}_s^0)^\top{\bm D}_{k,d}^{-1}(\widehat{\bm l}^{k,d}_t-{\bm H}^{k,d\top}{\bm l}_t^0)\big]^2\Big\}^{1/2}\|\bm v\|^2\\
 &\leq&\frac{1}{p}(\sum\limits_{s=1}^p\|\widehat{\bm l}^{k,d}_s-{\bm H}^{k,d\top}{\bm l}_s^0\|^2)\|{\bm D}^{-1}_{k,d}\|\cdot O_p(1)\\
 &=&O_p({m_{np}^{-1}});
\end{eqnarray*}
\begin{eqnarray*}
 III&=&\frac{1}{p}\sum\limits_{s=1}^p\sum\limits_{t=1}^px_{is}(\widehat{\bm l}^{k,d}_s-{\bm H}^{k,d\top}{\bm l}_s^0)^\top{\bm D}_{k,d}^{-1}{\bm H}^{k,d\top}{\bm l}_t^0x_{it}\\
  &\leq&\frac{1}{p}\Big\{\sum\limits_{s=1}^p\sum\limits_{t=1}^p\big[(\widehat{\bm l}^{k,d}_s-{\bm H}^{k,d\top}{\bm l}_s^0)^\top{\bm D}_{k,d}^{-1}{\bm H}^{k,d\top}{\bm l}_t^0\big]^2\Big\}^{1/2}\|\bm v\|^2\\
 &\leq&\frac{1}{p}(\sum\limits_{i=1}^p\|\widehat{\bm l}^{k,d}_s-{\bm H}^{k,d\top}{\bm l}_s^0\|^2)^{1/2}(\sum\limits_{s=1}^p\|{\bm H}^{k,d\top}{\bm l}_s^0\|^2)^{1/2}\|{\bm D}^{-1}_k\|\cdot O_p(1)\\
 &=&O_p(m_{np}^{-1/2}).
\end{eqnarray*}
Similarly,
\begin{equation*}IV=O_p(m_{np}^{-1/2}).\end{equation*}
The proof is then complete.
\end{proof}

\begin{lem}\label{lem3}
Suppose Assumptions \ref{ass1}--\ref{ass4} hold. For any  matrix ${\bm F}$, denote $W_i(d,{\bm F})=\frac{1}{p}{\bm x}_i^\top({\bm I}_d-\Pro_{\bm F}){\bm x}_i$. Then,  for $1\leq d\leq d_0$,
\begin{equation*}\frac{1}{n}\sum\limits_{k=1}^K\sum\limits_{i\in M_k}W_i(d,\widehat{\bm L}^{k,d})-\frac{1}{n}\sum\limits_{k=1}^K\sum\limits_{i\in M_k}W_i(d,{\bm L}^0{\bm H}^{k,d})=O_p(m_{np}^{-1/2}).
\end{equation*}
\end{lem}
\begin{proof}
Note that, $E\|{\bm x}_{i}\|^2=\sum\limits_{s=1}^p\E(x_{is}^2)\leq p\max\limits_{i,s}\E(x^2_{is})=O(p)$. According to Lemma \ref{lem2},
\begin{eqnarray*}
&&\frac{1}{n}\sum\limits_{k=1}^K\sum\limits_{i\in M_k}[W_i(d,\widehat{\bm L}^{k,d})-W_i(d,{\bm F}^0{\bm H}^{k,d})]\\
&=&\frac{1}{np}\sum\limits_{k=1}^K\sum\limits_{i\in M_k}{\bm x}_i^\top(\Pro_{{\bm L}^0{\bm H}^{k,d}}-\Pro_{\widehat{\bm L}^{k,d}}){\bm x}_i\\
&=&O_p(m_{np}^{-1/2}).
\end{eqnarray*}
\end{proof}

\begin{lem}\label{lem4}
Suppose Assumptions \ref{ass1}--\ref{ass4} hold. For any $d< d_0$, there exists a number $c_d>0$ such that,
\begin{equation*}\lim_{n,p\to\infty}\frac{1}{n}\sum\limits_{k=1}^K\sum\limits_{i\in M_k}[W_i(d,{\bm L}^0{\bm H}^{k,d})-W_i(d_0,{\bm L}^0)]>c_d.
\end{equation*}
\end{lem}
\begin{proof}
\begin{eqnarray*}
&&\frac{1}{n}\sum\limits_{k=1}^K\sum\limits_{i\in M_k}[W_i(d,{\bm L}^0{\bm H}^{k,d})-W_i(d_0,{\bm L}^0)]\\
&=&\frac{1}{np}\sum\limits_{k=1}^K\sum\limits_{i\in M_k}{\bm x}_i^\top(\Pro_{{\bm L}^0}-\Pro_{{\bm L}^0{\bm H}^{k,d}}){\bm x}_i\\
&=&\frac{1}{np}\sum\limits_{k=1}^K\sum\limits_{i\in M_k}{\bm e}_i^\top(\Pro_{{\bm L}^0}-\Pro_{{\bm L}^0{\bm H}^{k,d}}){\bm e}_i\\
&&+\frac{2}{np}\sum\limits_{k=1}^K\sum\limits_{i\in M_k}{\bm e}_i^\top(\Pro_{{\bm L}^0}-\Pro_{{\bm L}^0{\bm H}^{k,d}}){\bm L}^0{\bm f}_i^0\\
&&+\frac{1}{np}\sum\limits_{k=1}^K\sum\limits_{i\in M_k}{\bm f}_i^{0\top}{\bm L}^{0\top}(\Pro_{{\bm L}^0}-\Pro_{{\bm L}^0{\bm H}^{k,d}}){\bm L}^0{\bm f}_i^0\\
&=&I+II+III.
\end{eqnarray*}
We have $I\geq0$ since for each $d$, $(\Pro_{{\bm L}^0}-\Pro_{{\bm L}^0{\bm H}^{k,d}})$ is positive semi-definite. For II, note that for each $d$,
\begin{equation*}
\frac{2}{p}{\bm e}_i^\top(\Pro_{{\bm L}^0}-\Pro_{{\bm L}^0{\bm H}^{k,d}}){\bm L}^0{\bm f}_t^0=\frac{2}{p}[{\bm e}_i^\top{\bm L}^0{\bm f}_t^0-{\bm e}_i^\top\Pro_{{\bm L}^0{\bm H}^{k,d}}{\bm L}^{0}{\bm f}_t^0]=o_p(1),
\end{equation*}
since $\|\frac{1}{\sqrt{p}}{\bm L}^{0\top}{\bm e}_i\|=O_p(1)$.

For $III$, according to \cite{SW1998}, ${\bm H}^{k,d}$  converges to a matrix ${\bm H}_0^d$ with rank $d$, then
\begin{eqnarray*}
III&=&\frac{1}{np}\sum\limits_{k=1}^K\sum\limits_{i\in M_k}{\bm f}_i^{0\top}{\bm L}^{0\top}(\Pro_{{\bm L}^0}-\Pro_{{\bm L}^0{\bm H}^{k,d}}){\bm L}^0{\bm f}_i^0\\
&\to&\frac{1}{np}\sum\limits_{k=1}^K\sum\limits_{i\in M_k}{\bm f}_i^{0\top}{\bm L}^{0\top}(\Pro_{{\bm L}^0}-\Pro_{{\bm L}^0{\bm H}^{d}_0}){\bm L}^0{\bm f}_i^0\\
&=&\text{Tr}(\frac{1}{p}[{\bm L}^{0\top}{\bm L}^0-{\bm L}^{0\top}{\bm L}^0{\bm H}^{d}_0({\bm H}^{d\top}_{0}{\bm L}_0^\top{\bm L}_0{\bm H}^{d}_0)^{-1}{\bm H}^{d\top}_0{\bm L}^{0\top}{\bm L}^0]\cdot\frac{1}{n}{\bm F}^{0\top}{\bm F}^0)\\
&>&c_d>0,
\end{eqnarray*}
since ${\bm H}^{d}_0$ has rank $d$ and $\lambda_{d_0}(\frac{1}{n}{\bm F}^{0\top}{\bm F}^0)$ and $\lambda_{d_0}(\frac{1}{p}{\bm L}^{0\top}{\bm L}^0)$ are bounded away from $0$.
\end{proof}

\begin{lem}\label{lem5}
Suppose  Assumptions \ref{ass1}--\ref{ass4} hold. For $1\leq d_1 \leq d_0$ and $r<d_2\leq d_{\max}$,
\begin{equation*}
  \boldsymbol\phi^{\top}_{d_1}{\widehat{\boldsymbol\phi}}^k_{d_2}=O_{p}(m_{np}^{-1}).
  \end{equation*}
 \end{lem}
\begin{proof}
 Let $\widehat{\boldsymbol\Phi}^k_{d_0}=(\widehat{\boldsymbol\phi}^k_1,\widehat{\boldsymbol\phi}^k_2,...,\widehat{\boldsymbol\phi}^k_{d_0})$ and ${\boldsymbol\Phi}_{d_0}=(\boldsymbol\phi_1,\boldsymbol\phi_2,...,\boldsymbol\phi_{d_0})$. According to Davis-Kahan $\sin\Theta$ theorem \citep{DK1970,YWS2014}, there exist a orthogonal matrix $\widehat{\bm O}\in\R^{d_0\times d_0}$ such that
\begin{equation*}\|\widehat{\boldsymbol\Phi}^k_{d_0}\widehat{\bm O}-\boldsymbol\Phi_{d_0}\|\leq \frac{\sqrt{d_0}\|{\bm X}^\top{\bm X}-{\bm L}^0{\bm F}^{0\top}_{-M_k}{\bm F}^0_{-M_k}{\bm L}^{0\top}\|_{\text{op}}}{\lambda_{d_0}-{\lambda}_{d_0+1}},
\end{equation*}
where $\|\cdot\|_{\text{op}}$ is the operator norm.
Note that,
\begin{eqnarray*}
\|{\bm F}^{0\top}_{-M_k}{\bm E}_{-M_k}\|_{\text{op}}&=&\sup_{\|\boldsymbol\xi=1\|}\|{\bm F}^{0\top}_{-M_k}{\bm E}_{-M_k}\boldsymbol\xi\|\\
&\leq&\sup_{\|\boldsymbol\xi=1\|}\left[(\sum_{s=1}^p\|\sum_{i\notin M_k}{\bm f}_i^0e_{is}\|^2)^{1/2}\cdot \|\boldsymbol\xi\|\right]\\
&=&O_{p}(\sqrt{np}).
\end{eqnarray*}
Then,
\begin{eqnarray*}
&&\|{\bm X}^\top{\bm X}-{\bm L}^0{\bm F}^{0\top}_{-M_k}{\bm F}^0_{-M_k}{\bm L}^{0\top}\|_{\text{op}}\\
&=&\|\sum_{i\in M_k}{\bm x}_i{\bm x}_i^\top+{\bm L}^0{\bm F}^{0\top}_{-M_k}{\bm E}_{-M_k}+{\bm E}_{-M_k}{\bm F}^0_{-M_k}{\bm L}^{0\top}+{\bm E}^\top_{-M_k}{\bm E}_{-M_k}\|_{\text{op}}\\
&\leq&\|\sum_{i\in M_k}{\bm x}_i{\bm x}_i^\top\|_{\text{op}}+\|{\bm L}^0{\bm F}^{0\top}_{-M_k}{\bm E}_{-M_k}\|_{\text{op}}+\|{\bm E}_{-M_k}{\bm F}^0_{-M_k}{\bm L}^{0\top}\|_{\text{op}}+\|{\bm E}^\top_{-M_k}{\bm E}_{-M_k}\|_{\text{op}}\\
&\leq&\|\sum_{i\in M_k}{\bm x}_i{\bm x}_i^\top\|_{\text{op}}+2\|{\bm L}^0{\bm F}^{0\top}_{-M_k}\|_{\text{op}}\cdot\|{\bm E}_{-M_k}\|_{\text{op}}+\|{\bm E}^\top_{-M_k}{\bm E}_{-M_k}\|_{\text{op}}\\
&=&O_{p}(M_{np}).
\end{eqnarray*}
Consequently,
\begin{equation*}
\|\widehat{\boldsymbol\Phi}^k_{d_0}\widehat{\bm O}-\boldsymbol\Phi_{d_0}\|=O_{p}(\frac{M_{np}}{np})=O_{p}({m^{-1}_{np}}),
\end{equation*}
Let $\widehat{\bm o}_{d_1}$ be the transpose of the $d_1$-th row of $\widehat{\bm O}$, we can obtain that,
\begin{equation*}
\boldsymbol\phi^{\top}_{d_1}{\widehat{\boldsymbol\phi}}^k_{d_2}=\widehat{\boldsymbol\phi}^{k\top}_{d_2}({{\boldsymbol\phi}}_{d_1}-\widehat{\boldsymbol\Phi}_{d_0}\widehat{\bm o}_{d_1})+\widehat{\boldsymbol\phi}^{k\top}_{d_2}\widehat{\boldsymbol\Phi}_{d_0}\widehat{\bm o}_{d_1}=O_{p}({m^{-1}_{np}}).\end{equation*}
\end{proof}

\begin{lem}\label{lem6}
Suppose Assumptions \ref{ass1}--\ref{ass4} hold. For $d_0<d\leq d_{\max}$,
\begin{equation*}
\|\widehat{\boldsymbol\phi}_{d}^{k\top}{\bm L}^0{\bm f}^0_i\|^2=O_p(pm_{np}^{-2}).
\end{equation*}
\end{lem}
\begin{proof}
It is easy to see that ${\bm L}^0{\bm f}^0_i\in span\{\boldsymbol\phi_1,\boldsymbol\phi_2,...,\boldsymbol\phi_{d_0}\}$. Let ${\bm L}^0{\bm f}^0_i=\sum\limits_{s=1}^{d_0}b_s\boldsymbol\phi_s$ with $\sum\limits_{s=1}^{d_0}b_s^2=\|{\bm L}^0{\bm f}^0_i\|^2$. Note that,
\begin{eqnarray*}
\|\widehat{\boldsymbol\phi}_{d}^{k\top}{\bm L}^0{\bm f}^0_i\|^2&=&\|\sum\limits_{s=1}^{d_0}b_s\widehat{\boldsymbol\phi}_{d}^{k\top}\boldsymbol\phi_s\|^2\\
&\leq&\|{\bm L}^0{\bm f}^0_i\|^2\sum\limits_{s=1}^{d_0}\|\widehat{\boldsymbol\phi}_{d}^{k\top}\boldsymbol\phi_s\|^2\\
\end{eqnarray*}
Since $\|{\bm L}^0{\bm f}^0_i\|^2=O_p(p)$, according to Lemma \ref{lem5}, we have that,
\begin{equation*}
\|\widehat{\boldsymbol\phi}_{d}^{k\top}{\bm L}^0{\bm f}^0_i\|^2=O_p(pm_{np}^{-2}).
\end{equation*}
\end{proof}

\begin{lem}\label{lem7}
Suppose Assumptions \ref{ass1}--\ref{ass4} hold. Let $\text{\normalfont Cov}({\bm e}_i)=\boldsymbol\Sigma_{p,i}$, for $d_0<d\leq d_{\max}$,
\begin{equation*}
\frac{1}{n}\sum\limits_{k=1}^K\sum\limits_{i\in M_k}\sum\limits_{s=1}^pw^{k,d}_{s}(x_{is}-\widehat{\bm f}_i^{k,d\top}\widehat{\bm l}^{k,d}_s)^2=\frac{1}{n}\sum\limits_{k=1}^K\sum\limits_{i\in M_k} \text{\normalfont Tr}(\mathcal{Q}_{\widehat{\bm L}^{k,d}}\boldsymbol\Sigma_{p,i})+o_p(1).\\
\end{equation*}
\end{lem}
\begin{proof}
Let ${\bm H}^{k,d}$ be as defined in Lemma \ref{lem1} and ${\bm H}^{k,d-}$ be the  generalized inverse matrix of ${\bm H}^{k,d}$. Moreover, let $\mathcal{Q}_{\widehat{\bm L}^{k,d}}=\text{\normalfont diag}(w^{k,d}_1,w^{k,d}_2,...,w^{k,d}_p)$. It is easy to see that $\sup_s w_{s}^{k,d}=o_p(1)$.  Write ${\bm x}_i={\bm L}^0{\bm f}_i^0+{\bm e}_i=\widehat{\bm L}^{k,d}{\bm H}^{k,d-}{\bm f}_i^0+({\bm L}^0{\bm H}^{k,d}-\widehat{\bm L}^{k,d}){\bm H}^{k,d-}{\bm f}_i^0+{\bm e}_i$. Since $\|({\bm L}^0{\bm H}^{k,d}-\widehat{\bm L}^{k,d}){\bm H}^{k,d-}{\bm f}_i^0\|^2=O_p(1)$,

\begin{eqnarray*}
&&\frac{1}{n}\sum\limits_{k=1}^K\sum\limits_{i\in M_k}\sum\limits_{s=1}^pw^{k,d}_{s}(x_{is}-\widehat{\bm f}_i^{k,d\top}\widehat{\bm l}^{k,d}_s)^2\\
&=&\frac{1}{n}\sum\limits_{k=1}^K\sum\limits_{i\in M_k}({\bm x}_i-\widehat{\bm L}^{k,d\top}\widehat{\bm f}_i^{k,d})^\top \mathcal{Q}_{\widehat{\bm L}^{k,d}}({\bm x}_i-\widehat{\bm L}^{k,d\top}\widehat{\bm f}_i^{k,d})\\
&=&\frac{1}{n}\sum\limits_{k=1}^K\sum\limits_{i\in M_k}{\bm e}_i^\top \mathcal{Q}_{\widehat{\bm L}^{k,d}} {\bm e}_i+\frac{2}{n}\sum\limits_{k=1}^K\sum\limits_{i\in M_k}{\bm e}_i^\top \mathcal{Q}_{\widehat{\bm L}^{k,d}}({\bm L}^0{\bm H}^{k,d}-\widehat{\bm L}^{k,d}){\bm H}^{k,d-}{\bm f}_i^0\\
&&+ \frac{1}{n}\sum\limits_{k=1}^K\sum\limits_{i\in M_k}(({\bm L}^0{\bm H}^{k,d}-\widehat{\bm L}^{k,d}){\bm H}^{k,d-}{\bm f}_i^0)^\top \mathcal{Q}_{\widehat{\bm L}^{k,d}}({\bm L}^0{\bm H}^{k,d}-\widehat{\bm L}^{k,d}){\bm H}^{k,d-}{\bm f}_i^0\\
&=&\frac{1}{n}\sum\limits_{k=1}^K\sum\limits_{i\in M_k}{\bm e}_i^\top \mathcal{Q}_{\widehat{\bm L}^{k,d}} {\bm e}_i+o_p(1)\\
&=&\frac{1}{n}\sum\limits_{k=1}^K\sum\limits_{i\in M_k} \text{Tr}(\mathcal{Q}_{\widehat{\bm L}^{k,d}}\boldsymbol\Sigma_{p,i})+o_p(1).
\end{eqnarray*}
\end{proof}

\begin{lem}\label{lem8}
Suppose  Assumptions \ref{ass1}--\ref{ass4} hold.  Moreover, either Assumption \ref{ass5a} or \ref{ass5b} holds.  For $1\leq d_1\leq d_0$ and $d_0+1\leq d_2\leq d_{\max}$,
\begin{eqnarray*}
\frac{1}{\lambda_{d_1}-\lambda_{d_2}}=O_{a.s.}(\frac{1}{np}).
\end{eqnarray*}
Here, $\lambda$ can be replaced by  $\widehat\lambda$, $\lambda^k$, and $\widehat\lambda^k$.
\end{lem}
\begin{proof}
We only  prove the result for $\lambda_d$ and $\widehat{\lambda}_d$ here. The result for  $\lambda_d^k$ and $\widehat{\lambda}_d^k$ can be similarly verified.
For $\lambda_s$, as $\lambda_d=0$ for $d>d_0$ and $\lambda_d\geq \lambda_{d_0}$ for $d\leq d_0$, the result follows if we can show that,
 \begin{equation*}
 \lambda_{d_0}>Cnp \ \ \ a.s.,
 \end{equation*}
  with $C$ a positive constant. For any ${\bm \xi}\in\R^p$ with $\|{\bm \xi}\|^2=1$, we have
\begin{equation*}
 {\bm \xi}^\top{\bm L}^0{\bm F}^{0\top}{\bm F}^0{\bm L^{0\top}}{\bm \xi}\geq \lambda_{d_0}({\bm F}^{0\top}{\bm F}^0)\cdot\|{\bm L^{0\top}}{\bm \xi}\|^2=\lambda_{d_0}({\bm F}^{0\top}{\bm F}^0)\cdot{\bm \xi}^\top{\bm L}^0{\bm L}^{0\top}{\bm \xi}.
\end{equation*}
Then, according to Assumption \ref{ass1}, Assumption \ref{ass2} and  Min-Max theorem,
\begin{eqnarray*}
\lambda_{d_0}&=&\min_{dim({\bm U})=d_0}\Big\{\max_{\boldsymbol\xi\in{\bm U}, \|\boldsymbol\xi\|^2=1}\boldsymbol\xi^\top{\bm L}^0{\bm F}^{0\top}{\bm F}^0{\bm L}^{0\top}\boldsymbol\xi\Big\}\\
&\geq&\lambda_{d_0}({\bm F}^{0\top}{\bm F}^0)\cdot\min_{dim({\bm U})=d_0}\Big\{\max_{\boldsymbol\xi\in{\bm U}, \|\boldsymbol\xi\|^2=1}\boldsymbol\xi^\top{\bm L}^0{\bm L}^{0\top}\boldsymbol\xi\Big\}\\
&=&\lambda_{d_0}({\bm F}^{0\top}{\bm F}^0)\cdot\lambda_{d_0}({\bm L}^0{\bm L}^{0\top})\\
&\geq&Cnp \ \ \ a.s..
\end{eqnarray*}
Here $\lambda_{d_0}({\bm A})$ represent the $d_0$-th largest eigenvalue of ${\bm A}$.

For $\widehat{\lambda}_d$, note that ${\bm X}={\bm F}^0{\bm L}^{0\top}+{\bm E}$, by the Weyl's inequality for singular values (Theorem 3.3.16 of \cite{HJ1990}),
\begin{equation*}
|\sqrt{\widehat{\lambda}_d}- \sqrt{{\lambda}_d} |\leq \lambda^{\frac{1}{2}}_1({\bm E^\top \bm E}).
\end{equation*}
Under Assumption \ref{ass5a}, $\lambda_1(\frac{1}{n}{\bm E}^\top{\bm E})<2\sigma^2$. While  according to \cite{BPZ2015}, under Assumption \ref{ass5b}, $\lambda_1(\frac{1}{n}{\bm E}^\top{\bm E})$ is bounded. It follows that,
\begin{equation*}
\widehat{\lambda}_{d_0}>(\sqrt{Cnp}-\lambda^{\frac{1}{2}}_1({\bm E^\top \bm E}))^2>C'np\ \ \ a.s.
\end{equation*}
and
\begin{equation*}
\widehat\lambda_{d_0+1}=O_{a.s.}(\lambda_{1}({\bm E^\top\bm E}))=o_{a.s.}(np).
\end{equation*}
Then,
\begin{equation*}
\frac{1}{\widehat\lambda_{d_1}-\widehat\lambda_{d_2}}\leq \frac{1}{\widehat\lambda_{d_0}-\widehat\lambda_{d_0+1}}=O_{a.s.}(\frac{1}{np}).
\end{equation*}
\end{proof}

\begin{lem}\label{lem9}
Suppose  Assumptions \ref{ass1}--\ref{ass4} hold. Moreover, either Assumption \ref{ass5a} or \ref{ass5b} holds. For $d_0< d\leq d_{\max}$,
\begin{equation*}
\sum_{l=d_0+1}^d\sum\limits_{k=1}^K(\widehat{\lambda}_l-\widehat\lambda^k_l)\leq\sum_{l=d_0+1}^d\widehat{\lambda}_l
\end{equation*}
and
\begin{equation*}
\sum_{l=d_0+1}^d \sum\limits_{k=1}^K\sum\limits_{i\in M_k}\|{\bm x}_i^\top\widehat{\boldsymbol\phi}^k_l\|^2\leq \sum_{l=d_0+1}^d\widehat{\lambda}_l, \ \ a.s.
\end{equation*}
\end{lem}
\begin{proof}
Let $\widehat{\boldsymbol\phi}^k_l=\sum\limits_{s=1}^pc_{ls,k}\widehat{\boldsymbol\phi}_s$. Note that,
\begin{eqnarray*}
{\bm X^\top \bm X}\widehat{\boldsymbol\phi}^k_l={\bm X^\top\bm X}\sum_{s=1}^pc_{ls,k}\widehat{\boldsymbol\phi}_s=\sum_{s=1}^pc_{ls,k}\widehat\lambda_s\widehat{\boldsymbol\phi}_s,
\end{eqnarray*}
and
\begin{eqnarray*}
{\bm X^\top\bm X}\widehat{\boldsymbol\phi}^k_l&=&{\bm X}_{-M_k}^\top{\bm X}_{-M_k}\widehat{\boldsymbol\phi}^k_l+\sum_{i\in M_k}{\bm x}_i{\bm x}_i^\top\widehat{\boldsymbol\phi}^k_l\\
&=&\widehat{\lambda}^k_{l}\widehat{\boldsymbol\phi}^k_l+\sum_{i\in M_k}{\bm x}_i{\bm x}_i^\top\widehat{\boldsymbol\phi}^k_l\\
&=&\sum_{s=1}^pc_{ls,k}\widehat{\lambda}^k_{l}\widehat{\boldsymbol\phi}_s+\sum_{i\in M_k}{\bm x}_i{\bm x}_i^\top\widehat{\boldsymbol\phi}^k_l.
\end{eqnarray*}
We have
\begin{equation*}
\sum_{s=1}^pc_{ls,k}(\widehat{\lambda}_{s}-\widehat{\lambda}^k_{l})\widehat{\boldsymbol\phi}_s=\sum_{i\in M_k}{\bm x}_i{\bm x}_i^\top\widehat{\boldsymbol\phi}^k_l.
\end{equation*}
It follows that,
\begin{equation*}
\sum_{l=d_0+1}^d\sum_{s=1}^pc^2_{ls,k}(\widehat{\lambda}_{s}-\widehat{\lambda}^k_{l})^2=\sum_{l=d_0+1}^d\|\sum_{i\in M_k}{\bm x}_i{\bm x}_i^\top\widehat{\boldsymbol\phi}^k_l\|^2.
\end{equation*}
Since, for $1\leq s\leq d_0$ and $d_0< l\leq d_{\max}$, $$\frac{1}{\widehat{\lambda}_{s}-\widehat{\lambda}^k_{l}}=O_{a.s.}(\frac{1}{np}) \mbox{ and  }\|{\bm x}_i{\bm x}_i^\top\widehat{\boldsymbol\phi}^k_l\|^2\leq\|{\bm x}_i\|^2\cdot\|{\bm x}_i^\top\widehat{\boldsymbol\phi}^k_l\|^2,$$
we have,
\begin{equation*}
\sum_{l=d_0+1}^d\sum_{s=1}^{d_0}c^2_{ls,k}\widehat{\lambda}_{s}=O_{a.s.}(\sum_{l=d_0+1}^d\sum_{i\in M_k}\frac{\|{\bm x}_i\|^2}{np}\|{\bm x}_i^\top\widehat{\boldsymbol\phi}^k_l\|^2)=o_{a.s.}(\sum_{l=d_0+1}^d\sum_{i\in M_k}\|{\bm x}_i^\top\widehat{\boldsymbol\phi}^k_l\|^2),
\end{equation*}
Let $\widetilde{\bm L}^{d_0}=\sqrt{p}(\boldsymbol\phi_1,\boldsymbol\phi_2,...,\boldsymbol\phi_{d_0})$,
\begin{eqnarray*}
\sum_{l=d_0+1}^d\widehat\lambda_{l}&\geq&\sum_{l=d_0+1}^d\widehat{\boldsymbol\phi}^{k\top}_{l}(I_p-\Pro_{\widetilde{\bm L}^{d_0}}){\bm X^\top \bm X}(I_p-\Pro_{\widetilde{\bm L}^{d_0}})\widehat{\boldsymbol\phi}^k_{l}\\
&=&\sum_{l=d_0+1}^d\widehat{\boldsymbol\phi}^{k\top}_{l}{\bm X^\top \bm X}\widehat{\boldsymbol\phi}^k_{l}-\sum_{l=d_0+1}^d\widehat{\boldsymbol\phi}^{k\top}_{l}\Pro_{\widetilde{\bm L}^{d_0}}{\bm X^\top \bm X}\Pro_{\widetilde{\bm L}^{d_0}}\widehat{\boldsymbol\phi}^k_{l}\\
&=&\sum_{l=d_0+1}^d\widehat\lambda^{k}_{l}+\sum_{l=d_0+1}^d\sum_{i\in M_k} \|{\bm x}_i^\top\widehat{\boldsymbol\phi}^k_l\|^2-\sum_{l=d_0+1}^d\sum_{s=1}^{d_0}c^2_{ls,k}\widehat{\lambda}_{s}\\
&=&\sum_{l=d_0+1}^d\widehat\lambda^{k}_{l}+(1-o_{a.s.}(1))\sum_{l=d_0+1}^d \sum_{i\in M_k}\|{\bm x}_i^\top\widehat{\boldsymbol\phi}^k_l\|^2
\end{eqnarray*}
On the other hand, since ${(K-1)}{\bm X^\top \bm X}=\sum\limits_{k=1}^K{\bm X}_{-M_k}^\top{\bm X}_{-M_k}$, we have that,
\begin{eqnarray*}
(K-1)\sum_{l=d_0+1}^d\widehat\lambda_{l}&=&(K-1)\sum_{l=d_0+1}^d\widehat{\boldsymbol\phi}^\top_{l}{\bm X^\top \bm X}\widehat{\boldsymbol\phi}_{l}\\
&=&\sum\limits_{k=1}^K\sum_{l=d_0+1}^d\widehat{\boldsymbol\phi}^\top_{l}{\bm X}_{-M_k}^\top{\bm X}_{-M_k}\widehat{\boldsymbol\phi}_{l}\\
&=&\sum\limits_{k=1}^K\sum_{l=d_0+1}^d\widehat{\boldsymbol\phi}^\top_{l}({\bm I}_p-\Pro_{\widehat{\bm L}^{k,d_0}}){\bm X}_{-M_k}^\top{\bm X}_{-M_k}({\bm I}_p-\Pro_{\widehat{\bm L}^{k,d_0}})\widehat{\boldsymbol\phi}_{l}\\
&&+\sum\limits_{k=1}^K\sum_{l=d_0+1}^d\widehat{\boldsymbol\phi}^\top_{l}\Pro_{\widehat{\bm L}^{k,d_0}}{\bm X}_{-M_k}^\top{\bm X}_{-M_k}\Pro_{\widehat{\bm L}^{k,d_0}}\widehat{\boldsymbol\phi}_{l}\\
&=&\sum\limits_{k=1}^K\sum_{l=d_0+1}^d\widehat\lambda^{k}_{l}\cdot \|({\bm I}_p-\Pro_{\widehat{\bm L}^{k,d_0}})\widehat{\boldsymbol\phi}_{l}\|^2+\sum\limits_{k=1}^K\sum_{l=d_0+1}^d\sum_{t=1}^{d_0}\widehat{\lambda}^k_t\|\widehat{\boldsymbol\phi}^\top_{l}\widehat{\boldsymbol\phi}^k_{t}\|^2\\
&\leq&(1+o(1))\sum\limits_{k=1}^K\sum_{l=d_0+1}^d\widehat\lambda^{k}_{l}.
\end{eqnarray*}
Consequently,
\begin{equation*}
\sum_{l=d_0+1}^d\sum\limits_{k=1}^K(\widehat{\lambda}_l-\widehat\lambda^k_l)\leq (K-(K-1))\sum_{l=d_0+1}^d\widehat{\lambda}_l=\sum_{l=d_0+1}^d\widehat{\lambda}_l
\end{equation*}
and
\begin{equation*}
\sum_{l=d_0+1}^d \sum\limits_{k=1}^K\sum\limits_{i\in M_k}\|{\bm x}_i^\top\widehat{\boldsymbol\phi}^k_l\|^2\leq \sum_{l=d_0+1}^d\widehat{\lambda}_l, \ \ a.s.
\end{equation*}
\end{proof}

\begin{lem}\label{lem10}
Suppose  Assumptions \ref{ass1}--\ref{ass4} and \ref{ass5b} hold. Then, for $d_0< d\leq d_{\max}$,
 \begin{equation*}
\frac{1}{K}\sum\limits_{k=1}^K\E(1-|\widehat{\boldsymbol\phi}^\top_d\widehat{\boldsymbol\phi}^{k}_d|)^2=o(1).
\end{equation*}
\end{lem}
\begin{proof}
Here we only prove the result for $d=d_0+1$. For $d_0+1<d\leq d_{\max}$, the result can be similarly verified. Similar to Lemma \ref{lem9}, set $\widehat{\boldsymbol\phi}^k_d=\sum\limits_{s=1}^pc_{ds,k}\widehat{\boldsymbol\phi}_s$. Then
\begin{equation*}
\sum_{s=1}^pc^2_{ds,k}(\widehat{\lambda}_{s}-\widehat{\lambda}^k_{d})^2=\|\sum_{i\in M_k}{\bm x}_i{\bm x}_i^\top\widehat{\boldsymbol\phi}^k_d\|^2.
\end{equation*}
and
\begin{equation*}
\sum_{s=1}^pc^2_{ds,k}(\widehat{\lambda}_{s}-\widehat{\lambda}^k_{d})=\sum_{i\in M_k}\|{\bm x}_i^\top\widehat{\boldsymbol\phi}^k_d\|^2.
\end{equation*}
Write $\xi^{k,d}_1=\sum\limits_{s=1}^{d_0}c^2_{ds,k}$, $\xi^{k,d}_2= \sum\limits_{s=1}^{d_0}c^2_{ds,k}(\widehat{\lambda}_{s}-\widehat{\lambda}^k_{d})$. It is easy to see that,
\begin{equation*}
 \frac{1}{K}\sum\limits_{k=1}^K\xi^{k,d}_1=o_{a.s.}(1);
 \end{equation*}
 and, as shown in Lemma \ref{lem9},
 \begin{equation*}
\sum\limits_{k=1}^K\xi^{k,d}_2=o_{a.s.}(\sum\limits_{k=1}^K\|{\bm x}_i^\top\widehat{\boldsymbol\phi}^k_d\|^2)=o_{a.s.}(\widehat{\lambda}_d).
 \end{equation*}
 Note that,
\begin{eqnarray*}
&&c^2_{dd,k}(\widehat{\lambda}_{d}-\widehat{\lambda}^k_{d})+\xi^{k,d}_2\\
&=&\sum_{s=d+1}^pc^2_{ds,i}(\widehat{\lambda}^k_{d}-\widehat{\lambda}_{s})+\sum_{i\in M_k}\|{\bm x}_i^\top\widehat{\boldsymbol\phi}^k_d\|^2\\
&\geq&\sum_{s=d+1}^pc^2_{ds,t}(\widehat{\lambda}^k_{d}-\widehat{\lambda}_{d+1})\\
&=&(1-\xi^{k,d}_1-c^2_{dd,k})[(\widehat{\lambda}_{d}-\widehat{\lambda}_{d+1})-(\widehat{\lambda}_{d}-\widehat{\lambda}^k_{d})]\\
&=&(1-\xi^{k,d}_1-c^2_{dd,k})(\widehat{\lambda}_{d}-\widehat{\lambda}_{d+1})-(1-\xi^{k,d}_1)(\widehat{\lambda}_{d}-\widehat{\lambda}^k_{d})+c^2_{dd,k}(\widehat{\lambda}_{d}-\widehat{\lambda}^k_{d})\\
&\geq&(1-\xi^{k,d}_1-c^2_{dd,k})(\widehat{\lambda}_{d}-\widehat{\lambda}_{d+1})-(\widehat{\lambda}_{d}-\widehat{\lambda}^k_{d})+c^2_{dd,k}(\widehat{\lambda}_{d}-\widehat{\lambda}^k_{d}).
\end{eqnarray*}
We obtain that,
\begin{equation*}
\sum\limits_{k=1}^K\left[(1-\xi^{k,d}_1-c^2_{dd,k})(\widehat{\lambda}_{d}-\widehat{\lambda}_{d+1})-(\widehat{\lambda}_{d}-\widehat{\lambda}^k_{d})\right]\leq \sum\limits_{k=1}^K\xi_2^{k,d}.
\end{equation*}
Then,
\begin{equation*}
\frac{1}{K}\sum\limits_{k=1}^K(1-c^2_{dd,k})\leq \frac{\frac{1}{K}\sum\limits_{k=1}^K[\xi_2^{k,d}+(\widehat{\lambda}_d-\widehat{\lambda}^k_{d})]}{\widehat{\lambda}_d-\widehat{\lambda}_{d+1}}+\frac{1}{K}\sum\limits_{k=1}^K\xi_1^{k,d}.
\end{equation*}
According to Lemma \ref{lem9}, we have
\begin{equation*}
\sum\limits_{k=1}^K[\widehat{\lambda}_{d}-\widehat{\lambda}^k_{d}]\leq \widehat{\lambda}_{d}.
\end{equation*}
According to \cite{BPZ2015},  we have that, under Assumption \ref{ass5b}, $\widehat{\lambda}_d/{n}=O_p(1)$ and $\widehat{\lambda}_d-\widehat{\lambda}_{d+1}\sim O_p(n^{1/3})$. Furthermore, $1/K\leq \sup\limits_kn_k/n=o({n^{-2/3}})$. Then,
\begin{equation*}
\frac{1}{K}\sum\limits_{k=1}^K(1-{c}^2_{dd,k})=O_{a.s.}(\frac{\widehat{\lambda}_d}{K(\widehat{\lambda}_d-
\widehat{\lambda}_{d+1})})=o_p(1).
\end{equation*}
Note that $0\leq {c}^2_{dd,k}\leq 1$, we have $0\leq \frac{1}{K}\sum\limits_{k=1}^K(1-{c}^2_{dd,k})\leq 1$. Then,
\begin{equation*}
\lim_{t\to\infty}\sup_n \E\bigg(\Big|\frac{1}{K}\sum\limits_{k=1}^K(1-{c}^2_{dd,k})\Big|I_{|\frac{1}{K}\sum\limits_{k=1}^K(1-{c}^2_{dd,k})|>t}\bigg)=0.
\end{equation*}
According to Theorem 1.8 of \cite{Shao2003},
\begin{equation*}
\E(\frac{1}{K}\sum\limits_{k=1}^K(1-{c}^2_{dd,k}))=o(1).
\end{equation*}
Consequently,
\begin{eqnarray*}
\frac{1}{K}\sum\limits_{k=1}^K\E(1-|\widehat{\boldsymbol\phi}^\top_d\widehat{\boldsymbol\phi}^{k}_d|)^2&=&\frac{1}{K}\sum\limits_{k=1}^K\E(1-|c_{dd,k}|)^2\\
&\leq&\frac{1}{K}\sum\limits_{k=1}^K\E(1-|c_{dd,k}|)\\
&\leq&\frac{1}{K}\sum\limits_{k=1}^K\E(1-c^2_{dd,k})=o(1).
\end{eqnarray*}

\end{proof}

\begin{lem}\label{lem11}
Let ${\bm e}=(e_1,e_2,...,e_n)^\top$ be a random vector with $\max\limits_i\E( e_i^4)<C$ and ${\bm a}=(a_1,a_2,...,a_n)^\top$ a random vector that is independent with $\bm e$. Then,
\begin{equation*}
\E(\bm a^\top \bm e)^4\leq C\E\|{\bm a}\|^4.
\end{equation*}
\end{lem}
\begin{proof}
\begin{eqnarray*}
\E(\bm a^\top \bm e)^4&=&\E\left[\sum\limits_{i=1}^n\sum_{j=1}^n\sum_{k=1}^n\sum_{l=1}^n(a_{i}a_{j}a_{k}a_{l}e_{i}e_{j}e_{k}e_{l})\right]\\
&=&\sum\limits_{i=1}^n\sum_{j=1}^n\sum_{k=1}^n\sum_{l=1}^n\E(a_{i}a_{j}a_{k}a_{l})\E(e_{i}e_{j}e_{k}e_{l})\\
&\leq&\sum\limits_{i=1}^n\sum_{j=1}^n\sum_{k=1}^n\sum_{l=1}^n\E(a_{i}a_{j}a_{k}a_{l})\E(e_{i}^4)^{1/4}\E(e_{j}^4)^{1/4}\E(e_{k}^4)^{1/4}\E(e_{l}^4)^{1/4}\\
&\leq&C\sum\limits_{i=1}^n\sum_{j=1}^n\sum_{k=1}^n\sum_{l=1}^n\E(a_{i}a_{j}a_{k}a_{l})\\
&= &C\E[(\sum\limits_{i=1}^na_i)^4]\\
&\leq &C\E[(\sum\limits_{i=1}^na^2_i)^2]=C \E\|{\bm a}\|^4.
\end{eqnarray*}
\end{proof}

\begin{lem}\label{lem12}
Suppose  Assumptions \ref{ass1}--\ref{ass4} and Assumption \ref{ass5b} hold. For $d_0<d\leq d_{\max}$, there exist a constant $C_0$ such that
\begin{equation*}\widehat{\boldsymbol\phi}^{k\top}_d\boldsymbol\Sigma_p\widehat{\boldsymbol\phi}_d^{k}\to C_0 \ \ a.s..\end{equation*}
 \end{lem}
\begin{proof}
For $p\times p$ matrix $\bm A$, let $F^{\bm A}(t)=\frac{1}{p}\sum\limits_{s=1}^p{\mathds{1}}(\lambda_s(\bm A)\leq t)$ be its empirical spectral distribution (ESD).  The Stieltjes transform of a
nondecreasing function $F$ is defined by $m_F(z)=\int \frac{1}{t-z}dF(t)$ for all $z\in \mathbb{C}^{+}$ where $\mathbb{C}^{+}=\{z\in\mathbb{C}, Im(z)>0\}$. Let $\lim\limits_{n\to\infty}p/n=\rho$, by \cite{BZ2008}, we have, for all $k$,
$F^{\frac{1}{n-n_k}\bm E_{-M_k}^\top \bm E_{-M_k}}\to F \ a.s.$ with $m_F(z)$ satisfies $m_F(z)=\int \frac{1}{t(1-\rho-\rho zm_F(z))}dH(t)$.

  Let $\boldsymbol\Sigma_p=\sum\limits_{t=1}^p\tau_t{\bm v}_t{\bm v}_t^\top$ and
\begin{equation*}\Phi_p(\lambda,\tau)=\frac{1}{p}\sum_{s=1}^p\sum_{t=1}^p|\widehat{\boldsymbol\phi}^{k\top}_s {\bm v}_t|^2{\mathds{1}}(\lambda\geq \lambda^k_s){\mathds{1}}(\tau\geq \tau_t).
\end{equation*}
Then, by the procedures in \cite{LP2011},  \eqref{conv} holds if we can prove that  $|\text{Tr}(\boldsymbol\Sigma^s_p(\frac{1}{n-n_k}\bm X_{-M_k}^\top \bm X_{-M_k}-z{\bm I}_p)^{-1})-\text{Tr}(\boldsymbol\Sigma^s_p(\frac{1}{n-n_k}\bm E_{-M_k}^\top \bm E_{-M_k}-z {\bm I}_p)^{-1})|\to 0$ for $s=0,1,2,...$.
\begin{equation}\label{conv}\Phi_N(\lambda,\tau)\to \int_{-\infty}^\lambda\int_{-\infty}^\tau \phi(l,t)dH(l)dF(t), \ \ a.s. \end{equation}
where $\phi(l,t)$ is defined in \cite{LP2011}.

In fact, since $\lambda_1(\boldsymbol\Sigma_p^s)=\lambda_1^s(\boldsymbol\Sigma_p)<\infty$ and
\begin{eqnarray*}&&|\text{Tr}(\boldsymbol\Sigma^s_p(\frac{1}{n-n_k}\bm X_{-M_k}^\top \bm X_{-M_k}-z{\bm I}_p)^{-1})-\text{Tr}(\boldsymbol\Sigma^s_p(\frac{1}{n-n_k}\bm E_{-M_k}^\top \bm E_{-M_k}-z {\bm I}_p)^{-1})|\\
&\leq& \lambda_1(\boldsymbol\Sigma_p^s)|\text{Tr}((\frac{1}{n-n_k}\bm X_{-M_k}^\top \bm X_{-M_k}-z{\bm I}_p)^{-1})-\text{Tr}((\frac{1}{n-n_k}\bm E_{-M_k}^\top \bm E_{-M_k}-z {\bm I}_p)^{-1})|.
\end{eqnarray*}
we only need to prove that
\begin{equation*}|\text{Tr}((\frac{1}{n-n_k}\bm X_{-M_k}^\top \bm X_{-M_k}-z{\bm I}_p)^{-1})-\text{Tr}((\frac{1}{n-n_k}\bm E_{-M_k}^\top \bm E_{-M_k}-z {\bm I}_p)^{-1})|\to 0,
\end{equation*}
 which follows the facts that
 \begin{eqnarray*}
  &&F^{\frac{1}{n-n_k}\bm X_{-M_k}^\top \bm X_{-M_k}}=\text{Tr}((\frac{1}{n-n_k}\bm X_{-M_k}^\top \bm X_{-M_k}-z{\bm I}_p)^{-1});\\
  &&F^{\frac{1}{n-n_k}\bm E_{-M_k}^\top \bm E_{-M_k}}=\text{Tr}((\frac{1}{n-n_k}\bm E_{-M_k}^\top \bm E_{-M_k}-z{\bm I}_p)^{-1});\\
  &&F^{\frac{1}{n-n_k}\bm X_{-M_k}^\top \bm X_{-M_k}}-F^{\frac{1}{n-n_k}\bm E_{-M_k}^\top \bm E_{-M_k}}\leq \frac{1}{p}\text{rank}(\frac{1}{n-n_k}(\bm X_{-M_k}^\top \bm X_{-M_k}-\bm E_{-M_k}^\top \bm E_{-M_k}))\to 0.
 \end{eqnarray*}
 According to Weyl's inequality, $\widehat{\lambda}_{d+n_k}\leq \widehat{\lambda}^k_d\leq\widehat{\lambda}_d$. Moreover, by Proposition 1 of  \cite{Onatski2010},  exists a constant $C_\lambda$ such that $\lambda_d\to C_\lambda,\ \  a.s.$ for $d>d_0$ and $d/n\to0$. This implies that $\widehat{\lambda}^k_d\to C_\lambda, \ \ a.s. $ for $d_0<d\leq d_{\max}$. Let $(X_p,Y_p)$ follows $\Phi_p$, $(X,Y)$ follows $\Phi$ and $C_0=E(Y|X=C_{\lambda})$.
Then,
\begin{eqnarray*}
\widehat{\boldsymbol\phi}^{k\top}_d\boldsymbol\Sigma_p\widehat{\boldsymbol\phi}^{k}_d&=&\sum_{t=1}^p\tau_t|\widehat{\boldsymbol\phi}^{k\top}_d {\bm v}_t|^2\\
&=&\frac{\frac{1}{p}\sum\limits_{t=1}^p\tau_t|\widehat{\boldsymbol\phi}^{k\top}_d{\bm v}_t|^2}{\frac{1}{p}\sum\limits_{t=1}^p|\widehat{\boldsymbol\phi}^{k\top}_d {\bm v}_t|^2}\\
&=&\frac{\sum\limits_{t=1}^p\tau_tP(X_n=\lambda_d^k,Y_n=\tau_t)}{P(X_n=\lambda_d^k)}\\
&=&E(Y_n|X_n=\lambda_d^k)\\
&\to& E(Y|X=C_{\lambda})=C_0 \ \ a.s..
\end{eqnarray*}
\end{proof}

\begin{lem}\label{lem13}
Suppose  Assumptions \ref{ass1}--\ref{ass4} and Assumption \ref{ass5b} hold. For $d_0<d\leq d_{\max}$,
\begin{equation*}\text{\normalfont Var}(\frac{1}{n}\sum\limits_{k=1}^K\sum\limits_{i\in M_k}{\bm e}_{i}^\top\widehat{\boldsymbol\phi}^k_d\widehat{\boldsymbol\phi}_d^{k\top}{\bm e}_{i})\to 0.\end{equation*}
 \end{lem}

\begin{proof}
Similar to Lemma \ref{lem10}, we only prove the result for $d=d_0+1$.  Denote ${\bm X}_{-M_{k,l}}$ be the submatrix of ${\bm X}$ with the rows in $M_k$ and $M_l$ being removed, $\widehat{\lambda}_d^{k,l}$'s, and $\widehat{\boldsymbol\phi}_k^{k,l}$'s be the corresponding eigenvalues and eigenvectors.

Note that ${\bm e}_{i}$ is independent with $\widehat{\boldsymbol\phi}^k_d$ for $i\in M_k$, we have
\begin{equation*}\E[{\bm e}_{i}^\top\widehat{\boldsymbol\phi}^k_d\widehat{\boldsymbol\phi}_d^{k\top}{\bm e}_{i}]=\E[\widehat{\boldsymbol\phi}_d^{k\top}\boldsymbol\Sigma_p\widehat{\boldsymbol\phi}^k_d]=C_0+o(1).\end{equation*}
Then, let $i\neq j$ and $i\in M_k$, $j\in M_l$,
\begin{eqnarray*}
&&\text{Cov}[{\bm e}_{i}^\top\widehat{\boldsymbol\phi}^k_d\widehat{\boldsymbol\phi}_d^{k\top}{\bm e}_{i},{\bm e}_{j}^\top\widehat{\boldsymbol\phi}^l_d\widehat{\boldsymbol\phi}_d^{l\top}{\bm e}_{j}]\\
&=&\E[{\bm e}_{i}^\top\widehat{\boldsymbol\phi}^k_d\widehat{\boldsymbol\phi}_d^{k\top}{\bm e}_{i}{\bm e}_{j}^\top\widehat{\boldsymbol\phi}^l_d\widehat{\boldsymbol\phi}_d^{l\top}{\bm e}_{j}]-\E[{\bm e}_{i}^\top\widehat{\boldsymbol\phi}^k_d\widehat{\boldsymbol\phi}_d^{k\top}{\bm e}_{i}]\E[{\bm e}_{j}^\top\widehat{\boldsymbol\phi}^l_d\widehat{\boldsymbol\phi}_d^{l\top}{\bm e}_{j}]\\
&=&\E[{\bm e}_{i}^\top\widehat{\boldsymbol\phi}^k_d\widehat{\boldsymbol\phi}_d^{k\top}{\bm e}_{i}{\bm e}_{j}^\top\widehat{\boldsymbol\phi}^{k,l}_d\widehat{\boldsymbol\phi}_d^{k,l\top}{\bm e}_{j}]+\E[{\bm e}_{i}^\top\widehat{\boldsymbol\phi}^k_d\widehat{\boldsymbol\phi}_d^{k\top}{\bm e}_{i}({\bm e}_{j}^\top\widehat{\boldsymbol\phi}^l_d\widehat{\boldsymbol\phi}_d^{l\top}{\bm e}_{j}-{\bm e}_{j}^\top\widehat{\boldsymbol\phi}^{k,l}_d\widehat{\boldsymbol\phi}_d^{k,l\top}{\bm e}_{j})]-C_0^2+o(1)\\
&=&(I)+(II)-C_0^2+o(1).
\end{eqnarray*}
For $(I)$, we have
\begin{eqnarray*}
(I)&=&\E\{\E[{\bm e}_{i}^\top\widehat{\boldsymbol\phi}^k_d\widehat{\boldsymbol\phi}_d^{k\top}{\bm e}_{i}{\bm e}_{j}^\top\widehat{\boldsymbol\phi}^{k,l}_d\widehat{\boldsymbol\phi}_d^{k,l\top}{\bm e}_{j}|({\bm F}^0, {\bm L}^0, {\bm E}_{-M_k})]\}\\
&=&\E\{{\bm e}_{j}^\top\widehat{\boldsymbol\phi}^{k,l}_d\widehat{\boldsymbol\phi}_d^{k,l\top}{\bm e}_{j}\E[{\bm e}_{i}^\top\widehat{\boldsymbol\phi}^k_d\widehat{\boldsymbol\phi}_d^{k\top}{\bm e}_{i}|({\bm F}^0, {\bm L}^0, {\bm E}_{-M_k})]\}\\
&=&\E[{\bm e}_{j}^\top\widehat{\boldsymbol\phi}^{k,l}_d\widehat{\boldsymbol\phi}_d^{k,l\top}{\bm e}_{j}\widehat{\boldsymbol\phi}_d^{k\top}\boldsymbol\Sigma_p\widehat{\boldsymbol\phi}^k_d]\\
&=&C_0\E[\widehat{\boldsymbol\phi}_d^{k,l\top}\boldsymbol\Sigma_p\widehat{\boldsymbol\phi}^{k,l}_d]+o(1)\\
&=&C_0^2+o(1).
\end{eqnarray*}
For $(II)$, without loss of generality, we assume that $0\leq \widehat{\boldsymbol\phi}_d^{l\top}\widehat{\boldsymbol\phi}_d^{k,l}\leq 1$. Otherwise we use $-\widehat{\boldsymbol\phi}_d^{k,l}$ to replace $\widehat{\boldsymbol\phi}_d^{k,l}$. Then, $2\leq \|\widehat{\boldsymbol\phi}^l_d+\widehat{\boldsymbol\phi}_d^{k,l}\|^2\leq 4$.

Note that ${\bm e}_j$ is independent with $\widehat{\boldsymbol\phi}_d^{l}$ and $\widehat{\boldsymbol\phi}_d^{k,l}$, by Lemma \ref{lem11},
\begin{equation*}
\E\Big(({\widehat{\boldsymbol\phi}_d^{l}+\widehat{\boldsymbol\phi}_d^{k,l}})^\top{\bm e}_j\Big)^4\leq C\E\Big[\|\widehat{\boldsymbol\phi}_d^{l}+\widehat{\boldsymbol\phi}_d^{k,l})\|^4\Big]<\infty,
\end{equation*}
and
\begin{equation*}
\E\Big({\bm e}_j^\top({\widehat{\boldsymbol\phi}_d^{l}-\widehat{\boldsymbol\phi}_d^{k,l}})\Big)^4\leq C\E\Big[\|\widehat{\boldsymbol\phi}_d^{l}-\widehat{\boldsymbol\phi}_d^{k,l}\|\Big]^4.
\end{equation*}
It follows that,
\begin{eqnarray*}
(II)&\leq&\Big[\E({\bm e}_{i}^\top\widehat{\boldsymbol\phi}^k_d\widehat{\boldsymbol\phi}_d^{k\top}{\bm e}_{i})^2\Big]^{1/2}\Big[\E({\bm e}_{j}^\top\widehat{\boldsymbol\phi}^l_d\widehat{\boldsymbol\phi}_d^{l\top}{\bm e}_{j}-{\bm e}_{j}^\top\widehat{\boldsymbol\phi}^{k,l}_d\widehat{\boldsymbol\phi}_d^{k,l\top}{\bm e}_{j})^2\Big]^{1/2}\\
&\leq&\Big[\E({\bm e}_{i}^\top\widehat{\boldsymbol\phi}^k_d\widehat{\boldsymbol\phi}_d^{k\top}{\bm e}_{i})^2\Big]^{1/2}\bigg\{\E\Big[{\bm e}_{j}^\top(\widehat{\boldsymbol\phi}^l_d+\widehat{\boldsymbol\phi}_d^{k,l}){\bm e}_{j}^\top(\widehat{\boldsymbol\phi}^{l}_d-\widehat{\boldsymbol\phi}_d^{k,l})\Big]^2\bigg\}^{1/2}\\
&\leq&\Big[\E({\bm e}_{i}^\top\widehat{\boldsymbol\phi}^k_d\widehat{\boldsymbol\phi}_d^{k\top}{\bm e}_{i})^2\Big]^{1/2}\Big[\E({\bm e}_{j}^\top(\widehat{\boldsymbol\phi}^l_d+\widehat{\boldsymbol\phi}_d^{k,l}))^4\Big]^{1/4}\Big[\E({\bm e}_{j}^\top(\widehat{\boldsymbol\phi}^l_d-\widehat{\boldsymbol\phi}_d^{k,l}))^4\Big]^{1/4}\\
&\leq&C(\Big[\E\|\widehat{\boldsymbol\phi}^l_d-\widehat{\boldsymbol\phi}_d^{k,l}\|^4\Big]^{1/4}\\
&=&C'\Big[\E(1-|\widehat{\boldsymbol\phi}^{l\top}_d\widehat{\boldsymbol\phi}_d^{k,l}|)^2\Big]^{1/4}.
\end{eqnarray*}
Then,
\begin{eqnarray*}
&&\text{Var}(\frac{1}{n}\sum\limits_{k=1}^K\sum\limits_{i\in M_k}{\bm e}_{i}^\top\widehat{\boldsymbol\phi}^k_d\widehat{\boldsymbol\phi}_d^{k\top}{\bm e}_{i})\\
&=&\frac{1}{n^2}\sum\limits_{k=1}^K\sum\limits_{i\in M_k}\text{Var}({\bm e}_{i}^\top\boldsymbol\phi^k_d\boldsymbol\phi_d^{k\top}{\bm e}_{i})+\frac{1}{n^2}\sum\limits_{k=1}^K\sum\limits_{l=1}^K\sum\limits_{i\in M_k}\sum_{\substack{j\neq i\\j\in M_l
 }}\text{Cov}({\bm e}_{i}^\top\widehat{\boldsymbol\phi}^k_d\widehat{\boldsymbol\phi}_d^{k\top}{\bm e}_{i},{\bm e}_{j}^\top\widehat{\boldsymbol\phi}^l_d\widehat{\boldsymbol\phi}_d^{l\top}{\bm e}_{j})]\\
&\leq&\frac{C}{n}+\frac{C'}{n^2}\sum\limits_{k=1}^K\sum\limits_{l=1}^K\sum\limits_{i\in M_k}\sum_{\substack{j\neq i\\j\in M_l
 }}\Big[\E(1-|\widehat{\boldsymbol\phi}^{l\top}_d\widehat{\boldsymbol\phi}_d^{k,l}|)^2\Big]^{1/4}\\
 &=&\frac{C}{n}+\frac{C'}{n^2}\sum\limits_{k=1}^K\sum\limits_{l\neq k}\sum\limits_{i\in M_k}\sum_{j\in M_l
 }\Big[\E(1-|\widehat{\boldsymbol\phi}^{l\top}_d\widehat{\boldsymbol\phi}_d^{k,l}|)^2\Big]^{1/4}\\
  &=&\frac{C}{n}+\frac{C'}{n^2}\sum\limits_{k=1}^K\sum\limits_{l\neq k}n_kn_l\Big[\E(1-|\widehat{\boldsymbol\phi}^{l\top}_d\widehat{\boldsymbol\phi}_d^{k,l}|)^2\Big]^{1/4}\\
  &\leq&\frac{C}{n}+\frac{C'K\sup\limits_kn_k}{n^2}\sum\limits_{k=1}^Kn_k\frac{1}{K}\sum\limits_{l\neq k}\Big[\E(1-|\widehat{\boldsymbol\phi}^{l\top}_d\widehat{\boldsymbol\phi}_d^{k,l}|)^2\Big]^{1/4}\\
&=&\frac{C}{n}+o({C'K\sup\limits_kn_k/n})\\
&=&o(1).
\end{eqnarray*}
Here, the last last equalities follow from Lemma \ref{lem10} and Assumption \ref{ass4}.
\end{proof}

\begin{lem}\label{lem14}
Suppose  Assumptions \ref{ass1}--\ref{ass4} hold. Moreover, either Assumption \ref{ass5a} or \ref{ass5b} holds.  Then, for $d_0<d\leq d_{\max}$,
\begin{equation*}
\frac{1}{n}\sum_{l=d_0+1}^d\sum\limits_{k=1}^K\sum\limits_{i\in M_k}\|{\bm x}_i^\top\widehat{\boldsymbol\phi}^k_l\|^2< \frac{2}{n}\sum\limits_{k=1}^K\sum\limits_{i\in M_k} \text{\normalfont Tr}((\mathcal{Q}_{\widehat{\bm L}^{k,d}}-\mathcal{Q}_{\widehat{\bm L}^{k,d_0}})\boldsymbol\Sigma_{p,i})+o_p(1).
\end{equation*}
\end{lem}
\begin{proof}
On the one hand, when Assumption \ref{ass1} to \ref{ass3} and Assumption \ref{ass5a} hold, we have,
\begin{eqnarray*}
\text{Tr}((\mathcal{Q}_{\widehat{\bm L}^{k,d}}-\mathcal{Q}_{\widehat{\bm L}^{k,d_0}})\boldsymbol\Sigma_{p,i})&=&\sum_{s=1}^p(w^{k,d}_s-w^{k,d_0}_s)\sigma^2\\
&=&\sigma^2[\text{Tr}(\Pro_{\widehat{\bm L}^{k,d}})-\text{Tr}(\Pro_{\widehat{\bm L}^{k,d_0}})]\\
&=&(d-d_0)\sigma^2.
\end{eqnarray*}
While, according to Lemma \ref{lem9},
\begin{equation*}
\frac{1}{n}\sum_{l=d_0+1}^d \sum\limits_{k=1}^K\sum\limits_{i\in M_k}\|{\bm x}_i^\top\widehat{\boldsymbol\phi}^k_l\|^2\leq \frac{1}{n}\sum_{l=d_0+1}^d\widehat{\lambda}_l, \ \ a.s.
\end{equation*}
By Weyl's inequality, we have
\begin{equation*}
\frac{1}{n}\sum_{l=d_0+1}^d\widehat{\lambda}_l\leq \frac{d-d_0}{n}\lambda_{1}({\bm E^\top\bm E})<2(d-d_0)\sigma^2=  \frac{2}{n}\sum\limits_{k=1}^K\sum\limits_{i\in M_k} \text{Tr}((\mathcal{Q}_{\widehat{\bm L}^{k,d}}-\mathcal{Q}_{\widehat{\bm L}^{k,d_0}})\boldsymbol\Sigma_{p,i}) \ \ a.s..
\end{equation*}

On the other hand, when Assumption \ref{ass1} to \ref{ass3} and Assumption \ref{ass5b} hold, note that,
\begin{equation*}
\frac{1}{n}\sum_{l=d_0+1}^d \sum\limits_{k=1}^K\sum\limits_{i\in M_k}\|{\bm x}_i^\top\widehat{\boldsymbol\phi}^k_l\|^2=\frac{1}{n}\sum_{l=d_0+1}^d \sum\limits_{k=1}^K\sum\limits_{i\in M_k}\Big(\|{\bm e}_{i}^\top\widehat{\boldsymbol\phi}^k_l\|^2+2{\bm e}_{i}^\top\widehat{\boldsymbol\phi}^k_l\widehat{\boldsymbol\phi}^{k\top}_l{\bm L}^0{\bm f}^0_{i}+\|{{\bm f}_{i}^{0}}^\top{{\bm L}^{0}}^\top\widehat{\boldsymbol\phi}^k_l\|^2\Big).
\end{equation*}
Since $\E\|{\bm e}_{i}^\top\widehat{\boldsymbol\phi}^k_l\|^2=O(1)$ and $\E\|{{\bm f}_{i}^{0}}^\top{{\bm L}^{0}}^\top\widehat{\boldsymbol\phi}^k_l\|^2=O(pm_{np}^{-2})$, we have,
\begin{eqnarray*}
&&|\frac{1}{np}\sum\limits_{l=d_0+1}^d\sum\limits_{k=1}^K\sum\limits_{i\in M_k}{\bm e}_{i}^\top\widehat{\boldsymbol\phi}^k_l\widehat{\boldsymbol\phi}^{k\top}_l{\bm L}^0{\bm f}^0_{i}|\\
&\leq&\frac{1}{np}\sum\limits_{k=d_0+1}^d(\sum\limits_{k=1}^K\sum\limits_{i\in M_k}\|{\bm e}_{i}^\top\widehat{\boldsymbol\phi}^k_l\|^2)^{1/2}(\sum\limits_{k=1}^K\sum\limits_{i\in M_k}\|{{\bm f}_{i}^{0}}^\top{{\bm L}^{0}}^\top\widehat{\boldsymbol\phi}^k_l\|^2)^{1/2}\\
&=&O_p(p^{1/2}m_{np}^{-1})=o_p({1});
\end{eqnarray*}
and
\begin{equation*}
\frac{1}{n}\sum\limits_{l=d_0+1}^d\sum\limits_{k=1}^K\sum\limits_{i\in M_k}\|{{\bm f}_{i}^{0}}^\top{{\bm L}^{0}}^\top\widehat{\boldsymbol\phi}^k_l\|^2=O_{p}(pm^{-2}_{np})=o_p(1).
\end{equation*}
In addition, according to Lemma \ref{lem13}, for $d_0+1\leq l\leq d_{\max}$,
\begin{equation*}\E(\frac{1}{n}\sum\limits_{k=1}^K\sum\limits_{i\in M_k}{\bm e}_{i}^\top\widehat{\boldsymbol\phi}^k_l\widehat{\boldsymbol\phi}_l^{k\top}{\bm e}_{i})\to C_0;\end{equation*}
\begin{equation*}
\frac{1}{n}\sum\limits_{k=1}^K\sum\limits_{i\in M_k}\widehat{\boldsymbol\phi}_l^{k\top}\boldsymbol\Sigma_p\widehat{\boldsymbol\phi}^k_l\to C_0 \ \ \ a.s.;
\end{equation*}
and
\begin{equation*}\text{Var}(\frac{1}{n}\sum\limits_{k=1}^K\sum\limits_{i\in M_k}{\bm e}_{i}^\top\widehat{\boldsymbol\phi}^k_l\widehat{\boldsymbol\phi}_l^{k\top}{\bm e}_{i})\to 0.\end{equation*}
Thus
\begin{equation*}
\frac{1}{n}\sum\limits_{l=d_0+1}^d\sum\limits_{k=1}^K\sum\limits_{i\in M_k}\|{\bm e}_{i}^\top\widehat{\boldsymbol\phi}_l^{k}\|^2=\sum_{l=d_0+1}^dC_0+o_p(1).
\end{equation*}
Note that $\Pro_{\widehat{\bm L}^{k,d}}=\sum\limits_{l=1}^d\widehat{\boldsymbol\phi}_l^k\widehat{\boldsymbol\phi}_l^{k\top}$, we have,
\begin{equation*}
\frac{1}{n}\sum\limits_{k=1}^K\sum\limits_{i\in M_k} \text{Tr}((\mathcal{Q}_{\widehat{\bm L}^{k,d}}-\mathcal{Q}_{\widehat{\bm L}^{k,d_0}})\boldsymbol\Sigma_{p,i})=\frac{1}{n}\sum\limits_{k=1}^K\sum\limits_{i\in M_k}\sum_{l=d_0+1}^d\widehat{\boldsymbol\phi}_l^{k\top}\text{diag}(\boldsymbol\Sigma_p)\widehat{\boldsymbol\phi}_l^{k}.
\end{equation*}
Thus,
\begin{eqnarray*}
&&\frac{1}{n}\sum_{l=d_0+1}^d\sum\limits_{k=1}^K\sum\limits_{i\in M_k}\|{\bm x}_i^\top\widehat{\boldsymbol\phi}^k_l\|^2- \frac{2}{n}\sum\limits_{k=1}^K\sum\limits_{i\in M_k} \text{Tr}((\mathcal{Q}_{\widehat{\bm L}^{k,d}}-\mathcal{Q}_{\widehat{\bm L}^{k,d_0}})\boldsymbol\Sigma_{p,i})\\
&=&\widehat{\boldsymbol\phi}^{k\top}_l(\boldsymbol\Sigma_p-2\text{diag}(\boldsymbol\Sigma_p))\widehat{\boldsymbol\phi}^{k}_l+o_p(1)\\
&<&0+o_p(1),
\end{eqnarray*}
since $2\text{\normalfont{diag}}(\boldsymbol\Sigma_p)-\boldsymbol\Sigma_p$ is positive definite.
\end{proof}

\end{document}